\documentclass[11pt,draftclsnofoot,onecolumn]{IEEEtran}
\usepackage{mathrsfs}

\usepackage{amssymb,amsmath,color,graphicx, amsbsy, bm}
\usepackage{graphicx}
\usepackage{subfigure}
\usepackage{cite}
\usepackage{algorithmic}
\usepackage{amsmath}
\usepackage{amsfonts}
\usepackage{amssymb}
\usepackage{graphicx}
\usepackage{color}
\usepackage{setspace}
\usepackage{makecell}
\usepackage[boxed]{algorithm2e}

\newtheorem{dingli}{Theorem}
\newtheorem{lemma}{Lemma}
\newtheorem{remark}{Remark}
\IEEEaftertitletext{\vspace{-2ex}}%

\begin{document}

\title{\fontsize{21pt}{21pt}\selectfont{Low complexity sum rate maximization for single and multiple stream MIMO AF relay networks}}



\author{Cong~Sun,~\IEEEmembership{Student Member,~IEEE,}
        and~Eduard~Jorswieck,~\IEEEmembership{Senior Member,~IEEE}
\thanks{C. Sun is with the State Key Lab.\ of Scientific and Engineering Computing, ICMSEC, AMSS,
Chinese Academy of Sciences, Beijing, 100190,
China, e-mail: (suncong@lsec.cc.ac.cn).}
\thanks{E. Jorswieck is with Communications Theory, Communications Laboratory, Dresden University of Technology, Dresden, 01062, Germany, e-mail: (Eduard.Jorswieck@tu-dresden.de).}}

\thispagestyle{empty}
\maketitle \vspace{-2cm}

\thispagestyle{empty}
\begin{abstract}
A multiple-antenna amplify-and-forward two-hop interference network with multiple links and multiple relays is considered. We optimize transmit precoders, receive decoders and relay AF matrices to maximize the achievable sum rate. Under per user and total relay sum power constraints, we propose an efficient algorithm to maximize the total signal to total interference plus noise ratio (TSTINR). Computational complexity analysis shows that our proposed algorithm for TSTINR has lower complexity than the existing weighted minimum mean square error (WMMSE) algorithm. We analyze and confirm by simulations that the TSTINR, WMMSE and the total leakage interference plus noise (TLIN) minimization models with per user and total relay sum power constraints can only transmit a single data stream for each user. Thus we propose a novel multiple stream TSTINR model with requirement of orthogonal columns for precoders, in order to support multiple data streams and thus utilize higher Degrees of Freedom. Multiple data streams and larger multiplexing gains are guaranteed. Simulation results show that for single stream models, our TSTINR algorithm outperforms the TLIN algorithm generally and outperforms WMMSE in medium to high Signal-to-Noise-Ratio scenarios; the system sum rate significantly benefits from multiple data streams in medium to high SNR scenarios.
\end{abstract}

\newpage
\pagestyle{headings}
\setcounter{page}{1}
\section{Introduction}
Relays are often used to aid communications, not only to improve the Quality of Service (QoS) of the user pairs, which have weak direct links due to poor channel conditions, but also to increase the multiplexing gain of the network \cite{Madsen}. Among various relay transmit schemes,
the most effective ones are Amplify-and-Forward (AF), Compute-and-Forward (CF) and Decode-and-Forward (DF). Especially AF protocol is standardized as layer 1 relaying \cite{Iwamura}, and thus in popular research,
because of its simplicity and low complexity.

In this paper we consider the multiple link multiple relay network with non-regenerative relaying. There has been many works discussing the optimization of the relay beamforming weights. For single antenna case, \cite{Fazeli} and \cite{Sun} study models to solve
the optimal relay AF weights, where total relay transmit power is minimized under guaranteed Signal-to-Interference-plus-Noise-Ratio (SINR) requirements. There is also extension to multiple antenna case. In \cite{Li} the authors explore
the network with one multiple-antenna relay, and according to various relay AF matrix schemes, proposes ``IRC FlexCoBF" algorithm for transmit
beamforming matrices. For the networks with one user pair and parallel relays, \cite{Toding} discusses the joint optimization of source and relay
beamforming with different receiver filters. With the similar MIMO relay network as \cite{Toding}, \cite{Liu} investigates the optimal joint source and relay power allocation to maximize the end-to-end achievable rate. Extended to one transmitter, multiple receiver and multiple relay network, \cite{Choi} proposes a weighted mean square error minimization (WMMSE) model to solve the source and relay beamforming matrices with MMSE receiving filter. Recent work of \cite{Truong} is based on general MIMO
AF relay networks with multiple links and multiple relays. The authors provide algorithms to jointly optimize users' precoders, decoders and the relay AF
matrices. Total leakage interference plus noise
(TLIN) minimization and WMMSE models are proposed, both with per user and total relay transmit power constraints. The idea to construct the WMMSE models in \cite{Choi} and \cite{Truong} are similar. The WMMSE model in \cite{Truong} is also extended to that with individual user and individual relay power constraints. The precoders, decoders and the relay AF matrices
are solved alternatively, where each subproblem can achieve its optimal solution. However the algorithm has quite high computational complexity. Here we propose different approaches to approximate the system sum rate and derive lower complexity algorithm to solve the corresponding optimization problem.

The Degrees of Freedom (DoFs) of one network is closely related to its channel capacity. In high Signal-to-Noise-Ratio (SNR) scenarios, the capacity
increases linearly with the number of DoFs. The authors in \cite{Cadambe} propose a new technique of
Interference Alignment (IA), to maximize the achievable DoFs for MIMO networks. Such technique optimizes the precoders and decoders, in order to eliminate the network interference and approach the capacity of MIMO network. In MIMO networks, the IA technique has been deeply investigated \cite{Yetis,Razaviyayn,Gomadam2,Sun3}. It is shown in \cite{Wagner} that, with relays the achievable DoFs of the MIMO interference network are increased, and the capacity as well as the reliability are improved. In two-hop networks, \cite{Ning} studies the feasibility conditions and the algorithms for relay aided IA, restricted on single antenna case. Aiming to achieve the maximum DoFs of the $2\times2\times2$ MIMO relay network, \cite{Gou} and \cite{Vaze} study similar technique of aligned interference neutralization
to explore the optimal transmission scheme for single antenna and multiple antennas cases, respectively. \cite{Jeon1} investigates the ergodic capacity of a class of fading 2-user 2-hop networks with interference neutralization technique. In \cite{Gastpar} the maximum achievable DoFs for different kinds of MIMO interference channels and MIMO multiple hop networks are
listed and concluded. For the general $K\times R\times K$ MIMO relay network, the maximum DoFs are only analyzed with restriction to the number of relays. Interestingly, we observe by simulations
that the algorithms proposed in \cite{Truong} all lead to precoders with linearly dependent columns, which result in single transmit data stream corresponding to one DoF for each user, regardless of
the number of antennas at relay and user nodes. This is an impetus for us to propose multiple data stream models.

In our paper, we propose several models for the general MIMO relay network, according to different purposes and situations. The general transmit process and
system model are introduced in Section \MakeUppercase{\romannumeral 2}. In Section \MakeUppercase{\romannumeral 3} we set up a Total Signal
to Total Interference plus Noise Ratio (TSTINR) maximization model to approximate the system sum rate, with per user and total relay transmit power
constraints. Then this TSTINR model as well as the TLIN and WMMSE model in \cite{Truong} are extended to those with individual user and individual relay power
constraints. Also, the computational complexity of our algorithm is analyzed and compared with the WMMSE algorithm in \cite{Truong}. Our proposed algorithm is
shown to have lower complexity. Furthermore, to achieve more than one data streams for each user, we propose a
multiple stream TSTINR model in Section \MakeUppercase{\romannumeral 4}. Compared to the TSTINR model in Section \MakeUppercase{\romannumeral 3},
additional orthogonal constraints for precoders are added. In all the models, the precoding matrices, decoding matrices and relay beamforming matrices are
iterated alternatively. Each subproblem is efficiently solved, with sufficient reduction of the objective function in each iteration guaranteed. We provide simulation results in Section
\MakeUppercase{\romannumeral 5}. Since the network model in \cite{Truong} is more general than \cite{Choi}, we compare our proposed algorithms with those in \cite{Truong}. The results indicate that TSTINR outperforms TLIN generally, and achieves higher sum rate than
WMMSE in medium to high SNR scenarios, for the single stream cases. The system sum rate benefits much from the multiple stream model in medium and high SNR scenarios. Parts of our work are reported in \cite{Sun2} and \cite{Sun4}. Compared to them, we add more details of the proposed algorithms and provide detailed proof for all the mentioned theorems. Furthermore, we analyze and compare the detailed computational complexity of our proposed algorithm and the WMMSE algorithm from \cite{Truong}.

\emph{Notation}: Lowercase and uppercase boldface represent vectors
and matrices, respectively. $\mathbb{C}$ represents the complex
domain. $Re(a)$ means the real part of scalar $a$.
$\textrm{tr}(\mathbf{A})$ and $\|\mathbf{A}\|_{F}$ are the trace and
the Frobenius norm of matrix $\mathbf{A}$, respectively.
$\mathbf{I}_{d}$ represents the $d\times d$ identity matrix.
$\mathcal {K}$ and $\mathcal {R}$ represent the set of the user indices $\{1,2,\ldots,
K\}$ and that of relay indices $\{1,2,\ldots,
R\}$, respectively. And we use
$\mathbb{E}(\cdot)$ to denote the statistical expectation. $O(n)$ means the same order amount of
$n$. $\nu_{\min}^{d}(\mathbf{A})$ is composed of the eigenvectors of $\mathbf{A}$ corresponding to its $d$ smallest eigenvalues.

\section{System model}
Consider a two-hop interference channel consisting of $K$ user pairs and $R$ relays as in Fig. 1. Transmitter $k$, Receiver $k$ and Relay $r$ are equipped
with $M_{k}$, $N_{k}$ and $L_{r}$ antennas, respectively, for any $k\in\mathcal {K}$, $r\in\mathcal {R}$. User $k$ wishes to transmit $d_{k}$ parallel data streams. $\mathbf{s}_{k}\in \mathbb{C}^{d_{k}\times
1}$ denotes the transmit signal vector of User $k$, where $\mathbb{E}(\mathbf{s}_{k}\mathbf{s}_{k}^{H})=\mathbf{I}_{d_{k}}$. Due to the poor
channel conditions between user pairs, there is no direct links among users. Low-complex relays aid to
communicate and the AF transmit protocol is used. Here we assume perfect channel state
information (CSI) is available at a central controller.
\begin{figure}[htb]

\begin{minipage}[b]{1.0\linewidth}
  \centering
  \centerline{\includegraphics[width=12cm]{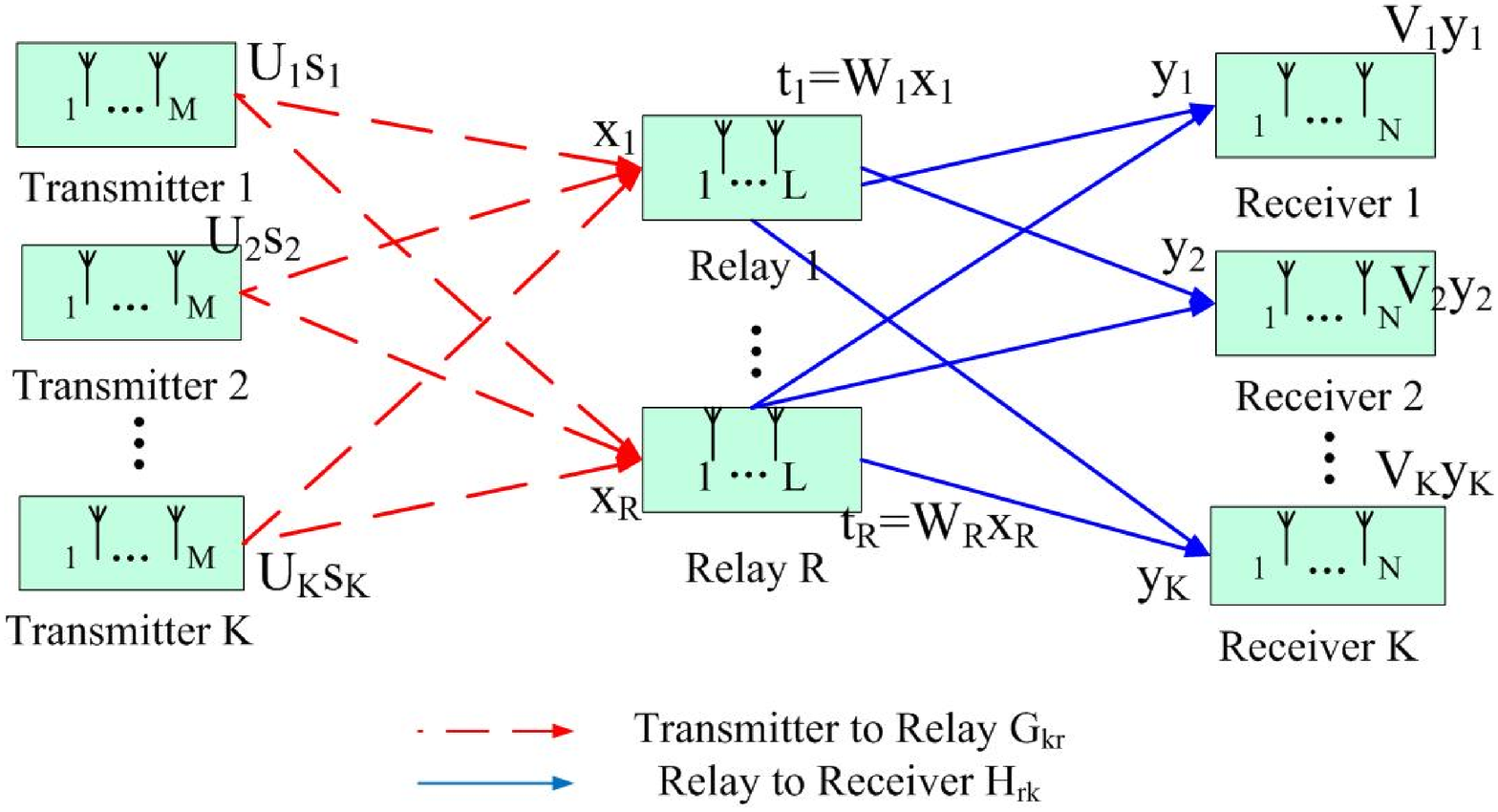}}
  \center{Fig. 1 MIMO relay AF network}\medskip
\end{minipage}
\end{figure}

Transmission process includes two time slots. In the first time slot, all sources transmit signals to all relays.
Relay $r$ receives $\mathbf{x}_{r}=\sum_{k\in\mathcal
{K}}\mathbf{G}_{rk}\mathbf{U}_{k}\mathbf{s}_{k}+\mathbf{n}_{r}$,
for all $r\in\mathcal{R}$, where $\mathbf{U}_{k}\in \mathbb{C}^{M_{k}\times
d_{k}}$ is the precoding matrix of User $k$, $\mathbf{G}_{rk}\in \mathbb{C}^{L_{r}\times
M_{k}}$ is the channel coefficient between the Transmitter $k$ and Relay $r$, and $\mathbf{n}_{r}$ with zero mean and variance matrix
$\sigma_{1}^{2}\mathbf{I}_{L_{r}}$ is the noise at Relay $r$. In the second time slot, by the AF protocol all relays broadcast to all
destinations $\mathbf{t}_{r}=\mathbf{W}_{r}\mathbf{x}_{r}$,
for all $r\in\mathcal {R}$,
where $\mathbf{W}_{r}\in \mathbb{C}^{L_{r}\times
L_{r}}$ is the beamforming matrix of Relay $r$.

Receiver $k$ observes:
$$
\mathbf{y}_{k}=\sum_{r\in\mathcal
{R}}\mathbf{H}_{kr}\mathbf{t}_{r}+\mathbf{z}_{k},
$$
for all $k\in\mathcal {K}$, where $\mathbf{H}_{kr}\in \mathbb{C}^{N_{k}\times
L_{r}}$ is the channel coefficient between Relay $r$ and Receiver $k$, and $\mathbf{z}_{k}$ with zero mean and variance matrix $\sigma_{2}^{2}\mathbf{I}_{N_{k}}$ is
the noise at Receiver $k$. Multiplying the decoding matrix $\mathbf{V}_{k}\in \mathbb{C}^{N_{k}\times
d_{k}}$, Receiver $k$ obtains:
\begin{eqnarray} \label{eq:eq5}
\tilde{\mathbf{y}}_{k}=\underbrace{\mathbf{V}_{k}^{H}\mathbf{T}_{kk}\mathbf{s}_{k}}_{\textrm{desired
signal}}+\underbrace{\sum_{q\in\mathcal {K},q\neq k}\mathbf{V}_{k}^{H}\mathbf{T}_{kq}\mathbf{s}_{q}}_{\textrm{interference}}+\underbrace{\sum_{r\in\mathcal {R}}\mathbf{V}_{k}^{H}\mathbf{H}_{kr}\mathbf{W}_{r}\mathbf{n}_{r}+\mathbf{V}_{k}^{H}\mathbf{z}_{k}}_{\textrm{noise}}.
\end{eqnarray}
The right hand side of (\ref{eq:eq5}) contains three terms: the desired
signal, the interference from other users and the noise including relay enhanced noise and the local noise. The effective channel from Transmitter $k$ to
Receiver $q$ is given by $\mathbf{T}_{kq}=\sum_{r\in\mathcal {R}}\mathbf{H}_{kr}\mathbf{W}_{r}\mathbf{G}_{rq}\mathbf{U}_{q}$. Suppose all the transmit
signals and noise in the system are independent of each other. The transmit powers at each user and each relay are, respectively:
\begin{eqnarray*}
P_{k}^{T}=\mathbb{E}(\|\mathbf{U}_{k}\mathbf{s}_{k}\|_{F}^{2})=\textrm{tr}(\mathbf{U}_{k}^{H}\mathbf{U}_{k}), k\in\mathcal {K},\\
P_{r}^{R}=\mathbb{E}(\|\mathbf{t}_{r}\|_{F}^{2})=\sum_{k\in\mathcal {K}}\|\mathbf{W}_{r}\mathbf{G}_{rk}\mathbf{U}_{k}\|_{F}^{2}
+\sigma_{1}^{2}\|\mathbf{W}_{r}\|_{F}^{2}, r\in\mathcal {R}.
\end{eqnarray*}
Then the total relay transmit power is $P^{R}=\sum_{r\in\mathcal {R}}P_{r}^{R}$.

In the following two sections we propose different models with corresponding algorithms to obtain efficient system precoders, decoders and relay
beamforming matrices. For the sake of expression simplicity, we predefine some symbols here: precoded and decoded effective channel from Transmitter $k$ to Relay $r$ as $\bar{\mathbf{G}}_{rk}=\mathbf{G}_{rk}\mathbf{U}_{k}$ and $\bar{\mathbf{W}}_{rk}=\mathbf{W}_{r}\mathbf{G}_{rk}$, respectively; precoded and decoded effective channel from Relay $r$ to Receiver $k$
$\bar{\mathbf{H}}_{kr}=\mathbf{H}_{kr}\mathbf{W}_{r}$ and $\bar{\mathbf{V}}_{kr}=\mathbf{V}_{k}^{H}\mathbf{H}_{kr}$, respectively
, $k\in\mathcal {K}, r\in\mathcal {R}$.

\section{Total Signal to Total Interference plus Noise Ratio model}
In this section, we develop a new model to approximate sum rate maximization. A low complexity algorithm to optimize the users' precoders, decoders and relay beamforming matrices is proposed. And its computational complexity is analyzed.

\subsection{A new model with per user and total relay power constraints} \label{sec:sec2}
First, we set up the new model of maximizing TSTINR with per user and total relay transmit
power constraints.

\subsubsection{Optimization problem formulation}
Define $\textrm{TSTINR}=\frac{P^{S}}{P^{I}+P^{N}}=\frac{\sum_{k\in\mathcal {K}}P_{k}^{S}}{\sum_{k\in\mathcal {K}}(P_{k}^{I}+P_{k}^{N})}$, where
\begin{eqnarray}
P_{k}^{S}&=&\mathbb{E}(\|\mathbf{V}_{k}^{H}\mathbf{T}_{kk}\mathbf{s}_{k}\|_{F}^{2})=\|\mathbf{V}_{k}^{H}\sum_{r\in\mathcal {R}}\mathbf{H}_{kr}\mathbf{W}_{r}\mathbf{G}_{rk}\mathbf{U}_{k}\|_{F}^{2},\\
P_{k}^{I}&=&\mathbb{E}(\|\sum_{q\in\mathcal {K},q\neq k}\mathbf{V}_{k}^{H}\mathbf{T}_{kq}\mathbf{s}_{q}\|_{F}^{2})=\sum_{q\in\mathcal {K},q\neq k}\|\mathbf{V}_{k}^{H}\sum_{r\in\mathcal {R}}\mathbf{H}_{kr}\mathbf{W}_{r}\mathbf{G}_{rq}\mathbf{U}_{q}\|_{F}^{2},\\
P_{k}^{N}&=&\mathbb{E}(\|\sum_{r\in\mathcal {R}}\mathbf{V}_{k}^{H}\mathbf{H}_{kr}\mathbf{W}_{r}\mathbf{n}_{r}
+\mathbf{V}_{k}^{H}\mathbf{z}_{k}\|_{F}^{2})=\sigma_{1}^{2}\sum_{r\in\mathcal {R}}\|\mathbf{V}_{k}^{H}\mathbf{H}_{kr}\mathbf{W}_{r}\|_{F}^{2}+\sigma_{2}^{2}\|\mathbf{V}_{k}\|_{F}^{2}
\end{eqnarray}
are the desired signal power, the leakage interference and the noise power at Receiver $k$, respectively.

We wish to maximize the system sum rate

\vspace{-1cm}
\begin{eqnarray} \label{eq:eq3}
R_{\textrm{sum}}=\frac{1}{2}\sum_{k\in\mathcal{K}}\textrm{log}_{2}\textrm{det}(\mathbf{I}_{N_{k}}+\mathbf{F}_{k}^{-1}\mathbf{T}_{kk}\mathbf{T}_{kk}^{H})
\end{eqnarray}
with $\mathbf{F}_{k}=\sum_{q\neq k,q\in\mathcal{K}}\mathbf{T}_{kq}\mathbf{T}_{kq}^{H}+\sum_{r\in\mathcal {R}}\bar{\mathbf{H}}_{kr}\bar{\mathbf{H}}_{kr}^{H}
+\sigma_{2}^{2}\mathbf{I}_{N_{k}}$. The direct optimization of the system sum rate is complicated. Therefore, we approximate it by the TSTINR and maximize the TSTINR instead. As the TSTINR remains
invariant with $\mathbf{V}_{k}$ replaced by $\mathbf{V}_{k}\mathbf{Q}$,
 where $\mathbf{Q}$ is any $d$-dimensional unitary matrix, we require the decoders $\mathbf{V}_{k}, k\in\mathcal {K}$ to be orthogonal, as the bases of the
 $d$-dimensional solution subspaces. Also, the following theorem holds:
\begin{dingli} \label{th:th1}
For any precoder $\mathbf{U}_{k},k\in\mathcal {K}$, relay beamforming matrix $\mathbf{W}_{r},r\in\mathcal {R}$ and any decoder $\mathbf{V}_{k}$ satisfying $\mathbf{V}_{k}^{H}\mathbf{V}_{k}=\mathbf{I}_{d_{k}}, k\in\mathcal {K}$, we have $\textrm{log}_{2}[1+\textrm{TSTINR}(\{\mathbf{U}\},\{\mathbf{V}\},\{\mathbf{W}\})]\leq R_{\textrm{sum}}(\{\mathbf{U}\},\{\mathbf{W}\})$.
\end{dingli}
The detailed proof is shown in Appendix-A. This states that the result from maximizing TSTINR provides a guaranteed system throughput.
Besides
the orthogonality constraints of decoders, we add fixed transmit power constraints\footnote{Fixed power constraints mean that all the power constraints are equality constraints, which are called constraints without power control in \cite{Truong}.} for per user and total relay. Then the corresponding
 optimization problem is:

\vspace{-0.8cm}
 \begin{subequations} \label{eq:eq4}
\begin{eqnarray}
\hspace{-0.8cm}&\displaystyle\max_{\substack{\{\mathbf{U}\},\{\mathbf{V}\},\\\{\mathbf{W}\}}}& \textrm{TSTINR}=
\frac{\sum_{k\in\mathcal {K}}P_{k}^{S}}{\sum_{k\in\mathcal {K}}(P_{k}^{I}+P_{k}^{N})}\label{eq:eq4.3}\\[-0.2cm]
\hspace{-0.8cm}&\textrm{s.t.}& \mathbf{V}_{k}^{H}\mathbf{V}_{k}=\mathbf{I}_{d_{k}},\\[-0.2cm]
\hspace{-0.8cm}&&\|\mathbf{U}_{k}\|_{F}^{2}=p_{0}^{T}, k\in\mathcal {K},\label{eq:eq4.2}\\[-0.2cm]
\hspace{-0.8cm}&&\sum_{r\in\mathcal {R}}\!(\sum_{k\in\mathcal {K}}\!\|\mathbf{W}_{r}\mathbf{G}_{rk}\mathbf{U}_{k}\|_{F}^{2}
+\sigma_{1}^{2}\|\mathbf{W}_{r}\|_{F}^{2})=p_{\max}^{R}.\label{eq:eq4.1}
\end{eqnarray}
\end{subequations}

\subsubsection{Problem reformulation}
There is a lack of efficient methods to deal with (\ref{eq:eq4.3}) because it is a fraction. This makes problem (\ref{eq:eq4}) difficult to solve. Stimulated by
Dinkelbach's work \cite{Dinkelbach} for nonlinear fraction optimization problem on convex sets, we use a parameter $C$ to combine the denominator and the
numerator as the new objective function, whereas the conclusions in \cite{Dinkelbach} cannot be extended to the problem (\ref{eq:eq4}) with nonconvex feasible set. Reformulate (\ref{eq:eq4}) as follows:
\begin{subequations} \label{eq:eq6}
\begin{eqnarray}
\hspace{-0.8cm}&\displaystyle\min_{\substack{\{\mathbf{U}\},\{\mathbf{V}\},\\\{\mathbf{W}\}}}& f(\{\mathbf{U}\},\{\mathbf{V}\},\{\mathbf{W}\};C)
=C(P^{I}+P^{N})-P^{S}=\sum_{k\in\mathcal {K}}[C(P_{k}^{I}+P_{k}^{N})-P_{k}^{S}]\label{eq:eq6.1}\\
\hspace{-0.8cm}&\textrm{s.t.}& \mathbf{V}_{k}^{H}\mathbf{V}_{k}=\mathbf{I}_{d_{k}},\label{eq:eq6.2}\\
\hspace{-0.8cm}&&\textrm{tr}(\mathbf{U}_{k}^{H}\mathbf{U}_{k})=p_{0}^{T}, k\in\mathcal {K},\label{eq:eq6.3}\\
\hspace{-0.9cm}&&\sum_{r\in\mathcal {R}}(\sum_{k\in\mathcal {K}}\|\mathbf{W}_{r}\mathbf{G}_{rk}\mathbf{U}_{k}\|_{F}^{2}
+\sigma_{1}^{2}\|\mathbf{W}_{r}\|_{F}^{2})=p_{\max}^{R}.\label{eq:eq6.4}
\end{eqnarray}
\end{subequations}
Thus in each iteration we solve (\ref{eq:eq6}), and then update the parameter $C$ as follows: initially $C$ is set as a small positive scalar
 (for example $C=1$), then after each iteration it is updated as
\begin{eqnarray} \label{eq:eq33}
C=\frac{P^{S}(\{\mathbf{U}\},\{\mathbf{V}\},\{\mathbf{W}\})}
{P^{I}(\{\mathbf{U}\},\{\mathbf{V}\},\{\mathbf{W}\})+P^{N}(\{\mathbf{U}\},\{\mathbf{V}\},\{\mathbf{W}\})}.
\end{eqnarray}

With such updating strategy of $C$, we have the following theorem, which is proved in Appendix-B:
\begin{dingli} \label{th:th2}
If the objective function of (\ref{eq:eq6}) has sufficient reduction in each iteration and $C$ is updated as (\ref{eq:eq33}), then the objective function of
(\ref{eq:eq4}), TSTINR, is monotonically increasing. Any stationary point of (\ref{eq:eq6}) is also a stationary point of (\ref{eq:eq4}).
\end{dingli}

\subsubsection{Alternating minimization algorithm}
The programming (\ref{eq:eq6}) itself is a nonconvex nonlinear matrix optimization problem, which is difficult to handle jointly.
Thus we solve precoders $\mathbf{U}_{k},k\in\mathcal {K}$, decoders $\mathbf{V}_{k},k\in\mathcal {K}$ and relay beamforming matrices
$\mathbf{W}_{r},r\in\mathcal {R}$ alternatively. Efficient algorithms are developed for each subproblem.

Firstly, we fix $\mathbf{U}_{k},k\in\mathcal {K}$ and $\mathbf{W}_{r},r\in\mathcal {R}$, then all $\mathbf{V}_{k},k\in\mathcal {K}$ are independent of
each other. The subproblem for $\mathbf{V}_{k}$ becomes:
\begin{eqnarray} \label{eq:eq7}
&\displaystyle\min_{\mathbf{X}\in\mathbb{C}^{N_{k}\times
d_{k}}}& \textrm{tr}(\mathbf{X}^{H}\mathbf{A}\mathbf{X})\nonumber\\
&\textrm{s.t.}& \mathbf{X}^{H}\mathbf{X}=\mathbf{I}_{d_{k}},
\end{eqnarray}
where $\mathbf{X}$ represents variable $\mathbf{V}_{k}$, and $\mathbf{A}=C\mathbf{F}_{k}-\mathbf{T}_{kk}\mathbf{T}_{kk}^{H}$.
Since $\mathbf{A}$ is Hermitian, we obtain the closed form solution of (\ref{eq:eq7}) as $\mathbf{X}=\nu_{\min}^{d}(\mathbf{A})$.

Next, we solve the subproblem for $\mathbf{W}_{r}$. Given a certain index $r\in\mathcal {R}$, we fix $\mathbf{U}_{k}$, $\mathbf{V}_{k},k\in\mathcal {K}$ and
$\{\mathbf{W}_{-r}\}$. Thus the optimization subproblem for $\mathbf{W}_{r}$ is:
\begin{eqnarray} \label{eq:eq9}
&\displaystyle\!\min_{\mathbf{X}\in\mathbb{C}^{L_{r}\times
L_{r}}}&\!\sum_{k\in\mathcal {K}}\textrm{tr}\big[\mathbf{X}(\mathbf{P}_{rr}^{k}+\sigma_{1}^{2}\mathbf{I}_{L_{r}})\mathbf{X}^{H}
\bar{\mathbf{V}}_{kr}^{H}\bar{\mathbf{V}}_{kr}\big]+2\textrm{Re}\big[\sum_{k\in\mathcal {K}}\!\sum_{\substack{l\neq r,\\l\in\mathcal {R}}}\!\textrm{tr}(\mathbf{X}\mathbf{P}_{rl}^{k}\mathbf{W}_{l}^{H}
\bar{\mathbf{V}}_{kl}^{H}\bar{\mathbf{V}}_{kr})\big]\nonumber\\
\!&\textrm{s.t.}&\!\! \textrm{tr}\big[\mathbf{X}(\sum_{k\in\mathcal {K}}\bar{\mathbf{G}}_{rk}\bar{\mathbf{G}}_{rk}^{H}
+\sigma_{1}^{2}\mathbf{I}_{L_{r}})\mathbf{X}^{H}\big]=\eta_{1},
\end{eqnarray}
where $\mathbf{P}_{rl}^{k}=C\sum_{q\neq k,q\in\mathcal{K}}\bar{\mathbf{G}}_{rq}\bar{\mathbf{G}}_{lq}^{H}
-\bar{\mathbf{G}}_{rk}\bar{\mathbf{G}}_{lk}^{H},k\in\mathcal {K}$, $r,l\in\mathcal{R}$ and $\eta_{1}=p_{\max}^{R}
-\sum_{l\neq r,l\in\mathcal {R}}\big(\sum_{k\in\mathcal {K}}\|\mathbf{W}_{l}\bar{\mathbf{G}}_{lk}\|_{F}^{2}+\sigma_{1}^{2}\|\mathbf{W}_{l}\|_{F}^{2}\big)$. Problem (\ref{eq:eq9}) is equivalent to a specific Quadratic Constrained Quadratic Programming (QCQP) with $\mathbf{x}=\textrm{vec}(\mathbf{X})$:
\begin{subequations}\label{eq:eq10}
\begin{eqnarray}
&\displaystyle\min_{\mathbf{x}\in\mathbb{C}^{L_{r}^2\times
1}} &\bar{f}(\mathbf{x})=\mathbf{x}^{H}\mathbf{B}_{1}\mathbf{x}+\mathbf{b}^{H}\mathbf{x}+\mathbf{x}^{H}\mathbf{b}\\
&\textrm{s.t.}& \mathbf{x}^{H}\mathbf{B}_{2}\mathbf{x}=\eta_{1}.\label{eq:eq10.1}
\end{eqnarray}
\end{subequations}
Here $\mathbf{B}_{1}=\sum_{k\in\mathcal {K}}(\mathbf{P}_{rr}^{k}+C\sigma_{1}^{2}\mathbf{I}_{L_{r}}
)^{T}\otimes(\bar{\mathbf{V}}_{kr}^{H}\bar{\mathbf{V}}_{kr})$, $\mathbf{B}_{2}=(\sum_{k\in\mathcal {K}}\bar{\mathbf{G}}_{rk}\bar{\mathbf{G}}_{rk}^{H}
+\sigma_{1}^{2}\mathbf{I}_{L_{r}})^{T}\otimes\mathbf{I}_{L}$ and $\mathbf{b}=\textrm{vec}(\sum_{k\in\mathcal {K}}\sum_{l\neq r,l\in\mathcal {R}}
\bar{\mathbf{V}}_{kr}^{H}\bar{\mathbf{V}}_{kl}\mathbf{W}_{l}\mathbf{P}_{rl}^{k})$. From the expressions, we know that $\mathbf{B}_{2}\succ0$ and generally
$\mathbf{B}_{1}$ is indefinite. Here we also discuss the case that $\mathbf{B}_{1}$ is positive semi-definite. First we show that the
following theorem holds:
\begin{dingli}
Given $\mathbf{B}_{2}=\mathbf{Q}^{H}\mathbf{Q}$, $\mathbf{Q}\succ0$, $\mathbf{p}=\mathbf{Q}\mathbf{x}$, $\bar{\mathbf{B}}_{1}
=\mathbf{Q}^{-1}\mathbf{B}_{1}\mathbf{Q}^{-1}$ and $\bar{\mathbf{b}}=\mathbf{Q}^{-1}\mathbf{b}$, (\ref{eq:eq10}) is equivalent to

\vspace{-2cm}
\begin{eqnarray} \label{eq:eq24}
\displaystyle\min_{\mathbf{p}^{H}\mathbf{p}=\eta_{1}} \mathbf{p}^{H}\bar{\mathbf{B}}_{1}\mathbf{p}+\bar{\mathbf{b}}^{H}\mathbf{p}
+\mathbf{p}^{H}\bar{\mathbf{b}}.
\end{eqnarray}
Further if $\mathbf{B}_{1}$ is indefinite, (\ref{eq:eq10}) is equivalent to:

\vspace{-1.5cm}
\begin{eqnarray} \label{eq:eq17}
\displaystyle\min_{\mathbf{p}^{H}\mathbf{p}\leq\eta_{1}} \mathbf{p}^{H}\bar{\mathbf{B}}_{1}\mathbf{p}+\bar{\mathbf{b}}^{H}\mathbf{p}
+\mathbf{p}^{H}\bar{\mathbf{b}}.
\end{eqnarray}
For the case of positive semi-definite $\mathbf{B}_{1}$: if the optimal solution of (\ref{eq:eq17}) $\mathbf{p}_{0}$ is not that of (\ref{eq:eq24}), then
$\mathbf{p}_{0}=(\bar{\mathbf{B}}_{1})^{-1}\bar{\mathbf{b}}$.
\end{dingli}
\begin{proof}
It is trivial to prove that (\ref{eq:eq10}) is equivalent to (\ref{eq:eq24}) and thus the detailed proof is omitted.

The global optimality conditions of (\ref{eq:eq17})
are as follows: there exists $\lambda\geq0$, such that $\bar{\mathbf{B}}_{1}+\lambda\mathbf{I}_{L_{r}^2}\succeq0$, $\mathbf{p}^{*}(\lambda)=(\bar{\mathbf{B}}_{1}+\lambda\mathbf{I}_{L_{r}^2})^{-1}\bar{\mathbf{b}}$, $\lambda(\|\mathbf{p}^{*}\|_{2}^2-\eta_{1})=0$ and
$\|\mathbf{p}^{*}(\lambda)\|_{2}^2\leq\eta_{1}$. For the case that $\mathbf{B}_{1}$ is indefinite, it follows that $\bar{\mathbf{B}}_{1}$ is also indefinite. Thus it must hold that $\lambda>0$. Then
from the complementary optimality condition we must have $\|\mathbf{p}^{*}(\lambda)\|_{2}^2=\eta_{1}$. Such $\lambda$ and $\mathbf{p}^{*}$ satisfy the global
optimal condition of (\ref{eq:eq24}). Therefore to solve (\ref{eq:eq10}) is equivalent to solve (\ref{eq:eq17}) with indefinite $\mathbf{B}_{1}$.

If the optimal solution of (\ref{eq:eq17}) $\mathbf{p}_{0}$ is not that of (\ref{eq:eq24}), then $\|\mathbf{p}_{0}\|^{2}<\eta_{1}$. From the complementary
optimality condition we have $\lambda=0$. Then $\mathbf{p}_{0}=(\bar{\mathbf{B}}_{1})^{-1}\bar{\mathbf{b}}$.
\end{proof}
Problem (\ref{eq:eq17}) is a typical trust region (TR) subproblem in trust region
optimization method. \cite[Chapter 6.1.1]{Yuan} provides an efficient
algorithm to achieve its optimal solution. It first checks whether $\|\mathbf{p}^{*}(0)\|^{2}\leq\eta_{1}$ with $\lambda=0$. If so, $\mathbf{p}^{*}(0)$ is the optimal
solution of (\ref{eq:eq17}); if not, the optimality conditions are used directly, and the optimal Lagrange multiplier $\lambda$ is calculated by
Newton's root-finding method from $\|\mathbf{p}^{*}(\lambda)\|_{2}^2=\eta_{1}$\footnote{From computation point of view, $\frac{1}{\|\mathbf{p}^{*}(\lambda)\|_{2}^2}=\frac{1}{\eta_{1}}$ is solved instead \cite{Yuan}.}. When $\mathbf{B}_{1}$ is indefinite, we solve (\ref{eq:eq17}) with the corresponding algorithm. When $\mathbf{B}_{1}$ is
positive semi-definite, we modify the algorithm to solve (\ref{eq:eq24}): check whether $\|\mathbf{p}^{*}(0)\|^{2}=\eta_{1}$, and if it is not the case, calculate $\lambda$ by
Newton's root-finding method. With this TR method we are able to solve (\ref{eq:eq10}) efficiently, and construct $\mathbf{W}_{r}$ from $\mathbf{x}$.

For a precoder $\mathbf{U}_{k}$, we fix $\mathbf{V}_{k},k\in\mathcal {K}$, $\mathbf{W}_{r},r\in\mathcal {R}$ and $\{\mathbf{U}_{-k}\}$ to get
the following subproblem.

\vspace{-0.7cm}
\begin{eqnarray} \label{eq:eq11}
&\displaystyle\min_{\mathbf{X}\in\mathbb{C}^{M_{k}\times
d_{k}}} &\textrm{tr}(\mathbf{X}^{H}\mathbf{Q}_{k}\mathbf{X})\nonumber\\
&\textrm{s.t.}& \|\mathbf{X}\|_{F}^{2}=p_{0}^{T},\nonumber\\
&&\textrm{tr}(\mathbf{X}^{H}\mathbf{L}_{k}\mathbf{X})=\eta_{2}.
\end{eqnarray}
\vspace{-0.2cm}
Here $\mathbf{X}$ represents $\mathbf{U}_{k},k\in\mathcal {K}$, and
\begin{eqnarray}
\mathbf{Q}_{k}\!\!=\!\!\sum_{r\in\mathcal {R}}\sum_{l\in\mathcal {R}}\bar{\mathbf{W}}_{rk}^{H}\left(C\!\!\sum_{q\neq k,q\in\mathcal{K}}\!\!
\bar{\mathbf{V}}_{qr}\bar{\mathbf{V}}_{ql}-\bar{\mathbf{V}}_{kr}\bar{\mathbf{V}}_{kl}\right)\bar{\mathbf{W}}_{lk}, \mathbf{L}_{k}=\sum_{r\in\mathcal {R}}\bar{\mathbf{W}}_{rk}^{H}\bar{\mathbf{W}}_{rk}, \nonumber\\
\eta_{2}=p_{\max}^{R}\!\!-\!\!\sum_{q\neq k,q\in\mathcal{K}}\sum_{r\in\mathcal {R}}\|\bar{\mathbf{W}}_{rq}\mathbf{U}_{q}\|_{F}^{2}
-\sigma_{1}^{2}\sum_{r\in\mathcal {R}}\|\mathbf{W}_{r}\|_{F}^{2}.\nonumber
\end{eqnarray}
Let $\mathbf{x}=\textrm{vec}(\mathbf{X})$, $\mathbf{C}_{1}=\mathbf{I}_{d_{k}}\otimes\mathbf{Q}_{k}$ and $\mathbf{C}_{2}=\mathbf{I}_{d_{k}}\otimes
\mathbf{L}_{k}$.
Then (\ref{eq:eq11}) is turned into a nonconvex QCQP:

\vspace{-0.7cm}
\begin{eqnarray} \label{eq:eq12}
&\displaystyle\min_{\mathbf{x}\in\mathbb{C}^{M_{k}d_{k}\times
1}} &\mathbf{x}^{H}\mathbf{C}_{1}\mathbf{x}\nonumber\\
&\textrm{s.t.}& \mathbf{x}^{H}\mathbf{x}=p_{0}^{T},\nonumber\\
&&\mathbf{x}^{H}\mathbf{C}_{2}\mathbf{x}=\eta_{2}.
\end{eqnarray}
\vspace{-0.2cm}
As all constraints are equalities, we use the Sequential Quadratic Programming (SQP) algorithm in \cite{Sun} to solve it. We set the initial point in the
SQP algorithm as the precoder calculated in the previous iteration. As SQP converges to a local optimal solution from the initial point, we
are able to guarantee the sufficient reduction of the objective function.

With the above analysis, we conclude the framework of the algorithm to solve (\ref{eq:eq4}):
\begin{algorithm}
\SetKwInOut{Input}{input}
\SetKwInOut{Output}{output}
\caption{Algorithm for single stream TSTINR model with total relay transmit power constraint}

\Input{initial value of $\mathbf{U}_{k},k\in\mathcal {K}$ and $\mathbf{W}_{r},r\in\mathcal {R}$, $C=1$}
\Output{$\mathbf{U}_{k},\mathbf{V}_{k},k\in\mathcal {K}$ and $\mathbf{W}_{r},r\in\mathcal {R}$}
\Repeat{Convergence}{
Update decoder $\mathbf{V}_{k}$ by solving (\ref{eq:eq7}), $k\in\mathcal {K}$\;
Update relay beamforming matrix $\mathbf{W}_{r}$ by solving
(\ref{eq:eq9}), $r\in\mathcal {R}$\;
Update precoder $\mathbf{U}_{k}$ by solving (\ref{eq:eq11}), $k\in\mathcal {K}$\;
Update $C$ as $C:=\frac{P^{S}}
{P^{I}+P^{N}}$\;
}
\end{algorithm}

\vspace{-0.8cm}
As we guarantee sufficient reduction of the objective function in each subproblem, the objective function value in our algorithm will converge. However as we have
separated the variables into more than two parts, there is no theoretical guarantee that the algorithm converges to a stationary point of (\ref{eq:eq4}).

\begin{remark} \label{re:re1}
Sharing similar expression of objective function, the algorithm here for TSTINR is also applicable to the TLIN model in \cite{Truong}.
The objective function of (\ref{eq:eq6}) is the linear combination of the total leakage interference plus noise $P^{I}+P^{N}$ and the
desired signal power $P^{S}$, and the parameter $C$ balances their weights, while the model TLIN in \cite{Truong} only minimizes $P^{I}+P^{N}$.
From the sum rate point of view, our TSTINR model is better motivated. This is verified by simulation results, where significant improvement
of system sum rate by TSTINR compared to TLIN is shown. A similar objective function has been discussed in \cite{Sun3}, where the desired signal power and the leakage interference are combined and optimized. In that case the leakage interference is aimed to be aligned perfectly, thus the parameter $C$ in \cite{Sun3} approaches to infinity to satisfy the interference alignment constraint. In our paper, $P^{I}+P^{N}$ might not be reduced to zero, and consequently $C$ might not grow to infinity. In \cite{Sun3}, $C$ is enlarged when the interference does not have sufficient reduction, which is different from the update strategy here.
\end{remark}

\subsection{Models with individual power constraints}
In this subsection, we extend our new TSTINR model, as well as the TLIN and WMMSE model in \cite{Truong}, to the ones with individual user and individual
relay fixed transmit power constraints.

\subsubsection{TSTINR model}
With individual user and individual relay fixed transmit power constraints, the TSTINR model becomes:

\vspace{-0.7cm}
\begin{eqnarray} \label{eq:eq13}
\hspace{-0.8cm}&\displaystyle\max_{\substack{\{\mathbf{U}\},\{\mathbf{V}\},\\\{\mathbf{W}\}}}&\hspace{-0.1cm} \textrm{TSTINR}
=\frac{\sum_{k\in\mathcal {K}}P_{k}^{S}}{\sum_{k\in\mathcal {K}}(P_{k}^{I}+P_{k}^{N})}\nonumber\\
\hspace{-0.8cm}&\textrm{s.t.}&\hspace{-0.1cm} \mathbf{V}_{k}^{H}\mathbf{V}_{k}=\mathbf{I}_{d_{k}},\nonumber\\
\hspace{-0.8cm}&&\hspace{-0.1cm}\|\mathbf{U}_{k}\|_{F}^{2}=p_{0}^{T}, k\in\mathcal {K},\nonumber\\
\hspace{-0.9cm}&&\hspace{-0.1cm}\sum_{k\in\mathcal {K}}\|\mathbf{W}_{r}\mathbf{G}_{rk}\mathbf{U}_{k}\|_{F}^{2}
+\sigma_{1}^{2}\|\mathbf{W}_{r}\|_{F}^{2}=p_{0}^{R},r\in\mathcal {R}.
\end{eqnarray}
\vspace{-0.2cm}
Assume there is a preprocess to carefully select active relays in the communication stage. Here we require all the users and relays to transmit signals with
fixed power. The difference from (\ref{eq:eq4}) is that, the relay sum power constraint is replaced by $R$ individual power constraints for each relay.

We use the objective function $f(\{\mathbf{U}\},\{\mathbf{V}\},\{\mathbf{W}\};C)$ from (\ref{eq:eq6}) in each iteration and preserve the same update strategy of
parameter $C$ as (\ref{eq:eq33}). When applying alternating iterations, the subproblems for decoders $\mathbf{V}_{k},k\in\mathcal {K}$ are the
same as (\ref{eq:eq7}). The objective functions in the subproblems for relay beamforming matrices $\mathbf{W}_{r},r\in\mathcal {R}$ and precoders
$\mathbf{U}_{k},k\in\mathcal {K}$ remain the same as in (\ref{eq:eq9}) and (\ref{eq:eq11}), respectively. The differences are the constraints.

With $\mathbf{X}$ represents $\mathbf{W}_{r}$, for any $r\in\mathcal {R}$, its constraint is:
\vspace{-0.2cm}
$$
\textrm{tr}\left[\mathbf{X}\left(\sum_{k\in\mathcal {K}}\bar{\mathbf{G}}_{rk}\bar{\mathbf{G}}_{rk}^{H}+\sigma_{1}^{2}\mathbf{I}_{L_{r}}\right)\mathbf{X}^{H}\right]=p_{0}^{R}.
$$
The transformed problem has the same structure as (\ref{eq:eq10}), and we solve it with the same method.

For a precoder $\mathbf{U}_{k}$, $k\in\mathcal {K}$, the constraints are:
\begin{eqnarray*}
\|\mathbf{X}\|_{F}^{2}&=&p_{0}^{T},\\
\textrm{tr}(\mathbf{X}^{H}\bar{\mathbf{W}}_{rk}^{H}\bar{\mathbf{W}}_{rk}\mathbf{X})&=&\eta_{3}^{r}, r\in\mathcal {R},
\end{eqnarray*}
where $\mathbf{X}$ here represents the variable $\mathbf{U}_{k}$, $\eta_{3}^{r}=p_{0}^{R}-\sum_{q\neq k,q\in\mathcal{K}}\|\bar{\mathbf{W}}_{rq}\mathbf{U}_{q}\|_{F}^{2}-\sigma_{1}^{2}\|\mathbf{W}_{r}\|_{F}^{2}$.

Then the subproblem is reformulated with $\mathbf{x}=\textrm{vec}(\mathbf{X})$:
\begin{eqnarray} \label{eq:eq14}
&\displaystyle\min_{\mathbf{x}\in\mathbb{C}^{M_{k}d_{k}\times
1}} &\mathbf{x}^{H}\mathbf{C}_{1}\mathbf{x}\nonumber\\
&\textrm{s.t.}& \mathbf{x}^{H}\mathbf{x}=p_{0}^{T},\nonumber\\
&&\mathbf{x}^{H}\mathbf{C}_{3}^{r}\mathbf{x}=\eta_{3}^{r}, r\in\mathcal {R},
\end{eqnarray}
\vspace{-0.2cm}
where $\mathbf{C}_{3}^{r}=\mathbf{I}_{d_{k}}\otimes(\bar{\mathbf{W}}_{rk}^{H}\bar{\mathbf{W}}_{rk})$.
Similar to solving (\ref{eq:eq12}), we achieve a local optimal solution of (\ref{eq:eq14}) and sufficient reduction of the objective function by SQP algorithm.

\begin{remark}
As the constraints in (\ref{eq:eq14}) are nonlinear equations, to ensure feasibility we normally require that the number of variables is no less than the number of equations.
Because we turn the variables from complex domain into real domain to solve them, the number increases to $2M_{k}d_{k}$. So we have $2M_{k}d_{k}\geq R+1$, for any $k\in\mathcal {K}$, as a requirement
for such problem. This limits our algorithm. The extension of our low
complexity algorithm to constraints with power control is ongoing and future work.
\end{remark}
\begin{remark}
The algorithm in \cite{Truong} cannot be extended to solve problems with individual relay fixed transmit power constraints, because it uses Semi-Definite Programming (SDP) relaxation
to solve the subproblems for precoders. If $R\geq3$, it may get a suboptimal solution by relaxation technique. Thus the objective function is not guaranteed to have sufficient reduction.
\end{remark}

With the adjustment to the subproblems, the basic framework of the algorithm is the same as that in the last subsection. As the objective function of TLIN in \cite{Truong} is similar to that of the reformulated problem of TSTINR, we extend the individual
relay power constraints case to TLIN with the corresponding replacement of the objective functions of each subproblem.

\subsubsection{WMMSE model}
Now we extend the individual fixed power constraint model to WMMSE model in \cite{Truong}. The corresponding optimization problem is as follows:

\vspace{-0.7cm}
\begin{eqnarray} \label{eq:eq15}
\hspace{-0.8cm}&\displaystyle\min_{\substack{\{\mathbf{U}\},\{\mathbf{V}\},\\\{\mathbf{W}\},\{\mathbf{S}\}}}& \sum_{k\in\mathcal {K}}
\big\{\textrm{tr}\big[\mathbf{S}_{k}(\mathbf{V}_{k}^{H}\bar{\mathbf{F}}_{k}\mathbf{V}_{k}-\mathbf{V}_{k}^{H}\mathbf{T}_{kk}
-\mathbf{T}_{kk}^{H}\mathbf{V}_{k}+\mathbf{I}_{d_{k}})\big]
-\textrm{log}_{2}\textrm{det}(\mathbf{S}_{k})\big\}\nonumber\\
\hspace{-0.8cm}&\textrm{s.t.}& \mathbf{S}_{k}\succeq0, \|\mathbf{U}_{k}\|_{F}^{2}=p_{0}^{T}, k\in\mathcal {K},\nonumber\\
\hspace{-0.8cm}&&\sum_{k\in\mathcal {K}}\|\mathbf{W}_{r}\mathbf{G}_{rk}\mathbf{U}_{k}\|_{F}^{2}
+\sigma_{1}^{2}\|\mathbf{W}_{r}\|_{F}^{2}=p_{0}^{R},r\in\mathcal {R},
\end{eqnarray}
\vspace{-0.2cm}
with $\bar{\mathbf{F}}_{k}=\mathbf{T}_{kk}\mathbf{T}_{kk}^{H}+\mathbf{F}_{k}$. And $\mathbf{S}_{k}\in\mathbb{C}^{d_{k}\times d_{k}}, k\in\mathcal {K}$ are the weight
matrices. In this approach, we try to
minimize the mean square error. \cite{Shi} shows that, if we use the linear MMSE receiver filter, (\ref{eq:eq15}) shares the same stationary points with
the sum rate maximization problem.

We also apply the alternating minimization algorithm to solve (\ref{eq:eq15}). From the above analysis, $\mathbf{V}_{k}$, for all $k\in\mathcal {K}$, are set as MMSE
filter. With fixed $\mathbf{U}_{k}$, for all $k\in\mathcal {K}$ and $\mathbf{W}_{r}$, for all $r\in\mathcal {R}$, we have
\vspace{-0.5cm}
\begin{eqnarray} \label{eq:eq19}
\mathbf{V}_{k}=\bar{\mathbf{F}}_{k}^{-1}\mathbf{T}_{kk}.
\end{eqnarray}
\vspace{-0.4cm}
Take the partial derivative of the objective function of (\ref{eq:eq15}) with respect to $\mathbf{S}_{k}$ and set the expression be zero. Then we obtain (\ref{eq:eq20}):

\vspace{-0.7cm}
\begin{eqnarray} \label{eq:eq20}
\mathbf{S}_{k}=[\mathbf{V}_{k}^{H}\bar{\mathbf{F}}_{k}\mathbf{V}_{k}
-\mathbf{V}_{k}^{H}\mathbf{T}_{kk}-\mathbf{T}_{kk}^{H}\mathbf{V}_{k}+\mathbf{I}_{d_{k}}]^{-1}=\mathbf{I}_{d_{k}}+\mathbf{T}_{kk}^{H}\mathbf{F}_{k}^{-1}\mathbf{T}_{kk}.
\end{eqnarray}
\vspace{-0.2cm}
Fixing all other variables, the subproblem for relay beamforming matrix $\mathbf{W}_{r}$, for any $r\in\mathcal {R}$, represented by $\mathbf{X}$ is expressed as follows:

\vspace{-0.7cm}
\begin{eqnarray} \label{eq:eq16}
\hspace{-0.8cm}&\displaystyle\min_{\mathbf{X}\in\mathbb{C}^{L_{r}\times
L_{r}}}& \sum_{k\in\mathcal {K}}\textrm{tr}\big[\mathbf{X}(\sum_{q\in\mathcal {K}}\bar{\mathbf{G}}_{rq}\bar{\mathbf{G}}_{rq}^{H}
+\sigma_{1}^{2}\mathbf{I}_{L_{r}})\mathbf{X}^{H}\bar{\mathbf{V}}_{kr}^{H}\mathbf{S}_{k}
\bar{\mathbf{V}}_{kr}\big]-2\textrm{Re}\big[\sum_{k\in\mathcal {K}}\!\sum_{q\in\mathcal {K}}\!\sum_{\substack{l\neq r,\\ l\in\mathcal{R}}}
\textrm{tr}(\mathbf{X}\bar{\mathbf{G}}_{rq}\bar{\mathbf{G}}_{lq}^{H}\mathbf{W}_{l}^{H}\bar{\mathbf{V}}_{kl}^{H}\mathbf{S}_{k}
\bar{\mathbf{V}}_{kr})\big]\nonumber\\
\hspace{-0.8cm}&\textrm{s.t.}&\textrm{tr}\big[\mathbf{X}(\sum_{k\in\mathcal {K}}\bar{\mathbf{G}}_{rk}\bar{\mathbf{G}}_{rk}^{H}
+\sigma_{1}^{2}\mathbf{I}_{L_{r}})\mathbf{X}^{H}\big]=p_{0}^{R}.
\end{eqnarray}
\vspace{-0.2cm}
By applying $\mathbf{x}=\textrm{vec}(\mathbf{X})$, (\ref{eq:eq16}) is transformed into a QCQP similar to (\ref{eq:eq10}). We use the same method to solve it.

For a precoder $\mathbf{U}_{k}$, for any $k\in\mathcal {K}$, while fixing all other variables, the subproblem becomes:

\vspace{-0.7cm}
\begin{eqnarray} \label{eq:eq21}
\!\!&\displaystyle\min_{\mathbf{X}\in\mathbb{C}^{M_{k}\times d_{k}}}&\!\! \textrm{tr}\left\{\mathbf{X}^{H}\big[\sum_{q\in\mathcal {K}}\big(\sum_{r\in\mathcal {R}}\bar{\mathbf{V}}_{qr}\bar{\mathbf{W}}_{rk}\big)^{H}\mathbf{S}_{q}(\sum_{l\in\mathcal {R}}\bar{\mathbf{V}}_{ql}\bar{\mathbf{W}}_{lk})\big]\mathbf{X}\right\}
-2\textrm{Re}\left[\textrm{tr}(\sum_{r\in\mathcal {R}}\mathbf{S}_{k}\bar{\mathbf{V}}_{kr}\bar{\mathbf{W}}_{rk}\mathbf{X})\right]\nonumber\\
&\textrm{s.t.}&\|\mathbf{X}\|_{F}^{2}=p_{0}^{T},\nonumber\\
&&\textrm{tr}(\mathbf{X}^{H}\bar{\mathbf{W}}_{rk}^{H}\bar{\mathbf{W}}_{rk}\mathbf{X})=\eta_{3}^{r}, r\in\mathcal {R},
\end{eqnarray}
\vspace{-0.2cm}
where $\mathbf{X}$ represents $\mathbf{U}_{k}$. With $\mathbf{x}=\textrm{vec}(\mathbf{X})$, we transform it into a QCQP similar to (\ref{eq:eq14}) and solve it efficiently by SQP method.

The algorithm to solve the WMMSE model with individual user and individual relay fixed power constraints is as follows:
\begin{algorithm}
\SetKwInOut{Input}{input}
\SetKwInOut{Output}{output}
\caption{Algorithm for WMMSE model with individual relay fixed power constraints}

\Input{initial value of $\mathbf{U}_{k},k\in\mathcal {K}$ and $\mathbf{W}_{r},r\in\mathcal {R}$}
\Output{$\mathbf{U}_{k},\mathbf{V}_{k},\mathbf{S}_{k},k\in\mathcal {K}$ and $\mathbf{W}_{r},r\in\mathcal {R}$}
\Repeat{Convergence}{
Update decoder $\mathbf{V}_{k}$ and weight matrix $\mathbf{S}_{k}$ by (\ref{eq:eq19}) and (\ref{eq:eq20}), $k\in\mathcal {K}$\;
Update relay
AF matrix $\mathbf{W}_{r}$ by solving (\ref{eq:eq16}), $r\in\mathcal {R}$\;
Update precoder $\mathbf{U}_{k}$ by solving (\ref{eq:eq21}), $k\in\mathcal {K}$\;
}
\end{algorithm}

\subsection{Computational complexity analysis}
In this subsection, we compare the computational complexity of the algorithm for our new model TSTINR with that of the algorithm in \cite{Truong} for
the WMMSE model\footnote{The algorithm for the WMMSE model and that for the TLIN model in \cite{Truong} are similar. Thus we only analyze WMMSE as a representative.}, both with per user and total relay transmit power as
representation. As both algorithms consist of three main parts of subproblems, we
analyze them individually.

First, we consider the complexity for solving decoder $\mathbf{V}_{k}$, as well as the weight matrix $\mathbf{S}_{k}$ in WMMSE, $k\in\mathcal {K}$. The
construction for the
 matrix $\mathbf{A}$ in (\ref{eq:eq7}) of TSTINR has the same computations as that for $\bar{\mathbf{F}}_{k}$ and $\mathbf{F}_{k}$ to solve $\mathbf{V}_{k}$ and $\mathbf{S}_{k}$ of WMMSE, similar to (\ref{eq:eq19}) and (\ref{eq:eq20}). Besides, TSTINR requires $9M^3$ operations for eigenvalue decomposition of $\mathbf{A}$; WMMSE requires $9M^3+9M^3=18M^3$
 operations for eigenvalue decompositions of $\bar{\mathbf{F}}_{k}$ and $\mathbf{F}_{k}$. Thus in the first part, TSTINR requires less complexity than WMMSE.

Second, we focus on the part for the relay beamforming matrix $\mathbf{W}_{r}$, with any $r\in\mathcal {R}$. Both algorithms require similar
complexity to construct the corresponding subproblem. To solve subproblem (\ref{eq:eq10}), our TSTINR algorithm applies the TR method, which
is mainly the Newton's root finding method for $\lambda$. It requires only a few inner iterations to find the optimal $\lambda$ and in each
inner iteration the main calculation is the QR factorization, with complexity of $O((L^2)^3)=O(L^6)$. To the contrary, the WMMSE algorithm in \cite{Truong}
applies SDP method, whose complexity is $O((L^2)^6)=O(L^{12})$. The complexity in the second part of TSTINR is much less than that of WMMSE.

Third is the part for the precoder $\mathbf{U}_{k}$, with any $k\in\mathcal {K}$. TSTINR solves (\ref{eq:eq12}) by SQP method, which complexity
is $O((Md)^3)=O(M^{3}d^3)$. WMMSE solves a QCQP with the same structure but with inequality constraints. And it applies SDP relaxation method, with
complexity of $O((Md)^6)=O(M^{6}d^6)$. With similar computations to construct the corresponding subproblem, TSTINR has much lower compelxity than
WMMSE in the third part.
\begin{table}[!hbp]
\begin{center}
\begin{tabular}{|l|c|c|}
\hline
Complexity comparison for each subproblem in one iteration             & TSTINR & WMMSE\\
\hline
\hline
1. $\mathbf{V}_{k}$ and $\mathbf{S}_{k}$, for any $k\in\mathcal {K}$  & $9M^3$ & $18M^3$\\
\hline
2. $\mathbf{W}_{r}$, for any $r\in\mathcal {R}$                       & $O(L^6)$ & $O(L^{12})$\\
\hline
3. $\mathbf{U}_{k}$, for any $k\in\mathcal {K}$                       & $O(M^{3}d^3)$ & $O(M^{6}d^6)$\\
\hline
\end{tabular}
\caption{Computational complexity analysis for each subproblem in one iteration}
\end{center}
 \end{table}

\vspace{-0.8cm}
The complexity differences of each subproblem in one iteration between the TSTINR and the WMMSE model are listed in Table 1. From the comparison of the three main parts of the two algorithms, we conclude that our proposed new algorithm for TSTINR enjoys lower complexity than the
algorithm for WMMSE in \cite{Truong}. Similarly, it is analyzed that our new algorithm for TSTINR\footnote{Our proposed algorithm is also applicable to the TLIN model,
which has
similar complexity as that for TSTINR.} has lower complexity than the algorithm for TLIN in \cite{Truong}. With numerical evidence, it turns out that the compared algorithms have similar number of iterations to solve the problem.

\begin{remark}
Because the SQP method, which solves (\ref{eq:eq12}), (\ref{eq:eq14}) and (\ref{eq:eq21}), can only deal with equality constraints, we restrict our models to fixed power constraints.
Although the models with power control constraints are more general and practical, the application of SQP method saves computational complexity.
Even with fixed power constraints, our TSTINR model performs better than the WMMSE model with power control constraints in medium to high SNR scenarios,
as shown in Section \ref{sec:sec3}.
\end{remark}

\section{Multiple stream model}
In this section, we study the multiple data stream model, with the purpose to maximize the system sum rate. It is pointed out in \cite{Yu} that multiple data streams, corresponding to multiple DoFs, help to increase the capacity of single-hop network in medium to high SNR. This conclusion can be extended to the two-hop case, by treating the network as an equivalent single-hop network between users and assuming the same power at the relays and the transmitters. First, we analyze the achievable number of data streams of the models from Section \MakeUppercase{\romannumeral 3}. Then our proposed TSTINR model
and the corresponding algorithm are modified
to support multiple data streams for each user pair. Here all the models include per user and total relay transmit power constraints.

\subsection{Analysis of single stream models} \label{sec:sec1}
The dimension $d_{k}$ of the transmit signal $\mathbf{s}_{k}$, is expected as the achieved number of data streams at User $k$. However, in simulations when
we apply our TSTINR algorithm to the system with $d_{k}>1$, we always observe that the system precoder $\mathbf{U}_{k}$ has rank one.
This implies that with linearly dependent columns of each precoder we can only achieve one data stream for each user pair, regardless of $d_{k}$. Similar phenomena are observed for the TLIN and WMMSE
algorithm in \cite{Truong}.

The following theorem provides theoretical evidence for the phenomena of the TSTINR and TLIN models:
\begin{dingli} \label{th:th3}
In our proposed TSTINR model and in the TLIN model from \cite{Truong}, the subproblem for precoder $\mathbf{U}_{k}$ always has a rank one
optimal solution, regardless of $d_{k}$.
\end{dingli}
\begin{proof}
Define $\mathbf{Y}=\mathbf{X}\mathbf{X}^{H}$ and drop the rank constraint of $\textrm{rank}(\mathbf{Y})\leq d_{k}$, the subproblem (\ref{eq:eq11}) in TSTINR is
relaxed
to the following semi-definite programming:
\begin{eqnarray} \label{eq:eq29}
&\displaystyle\min_{\mathbf{Y}\succeq0}&\textrm{tr}(\mathbf{YQ}_{k})\nonumber\\
&\textrm{s.t.}& \textrm{tr}(\mathbf{Y})=p_{0}^{T}, \textrm{tr}(\mathbf{L}_{k}\mathbf{Y})=\eta_{2}.
\end{eqnarray}
In \cite[Theorem 4.1]{Huang}, it is shown that (\ref{eq:eq29}) always has a rank one optimal solution. Suppose $\mathbf{Y}^{*}=\mathbf{y}^{*}(\mathbf{y}^{*})^{H}$ is the
optimal solution of (\ref{eq:eq29}). Let $\mathbf{X}^{*}=[a_{1}\mathbf{y}^{*},a_{2}\mathbf{y}^{*},\ldots,a_{d}\mathbf{y}^{*}]$ with
$a_{i}\in\mathbb{R}, \sum_{i=1}^{d_{k}}a_{i}^{2}=1$. That is, the columns of $\mathbf{X}^{*}$ consist of $a_{i}\mathbf{y}^{*}, i=1,\ldots,d_{k}$. From the fact that
$\textrm{tr}\big[(\mathbf{X}^{*})^{H}\mathbf{Q}_{k}\mathbf{X}^{*}\big]=\sum_{i=1}^{d_{k}}a_{i}^{2}(\mathbf{y}^{*})^{H}\mathbf{Q}_{k}\mathbf{y}^{*}
=\textrm{tr}(\mathbf{Y}^{*}\mathbf{Q}_{k})$,
for any feasible $\mathbf{X}$ of (\ref{eq:eq11}), we have:
$$
\textrm{tr}(\mathbf{X}^{H}\mathbf{Q}_{k}\mathbf{X})\geq\textrm{tr}(\mathbf{Y}^{*}\mathbf{Q}_{k})=
\textrm{tr}\big[(\mathbf{X}^{*})^{H}\mathbf{Q}_{k}\mathbf{X}^{*}\big].
$$
\vspace{-0.2cm}
Thus we conclude $\mathbf{X}^{*}$ is an optimal solution of (\ref{eq:eq11}). Because the subproblem for precoder in TLIN from \cite{Truong} has the same structure as
(\ref{eq:eq11}), we conclude the same result for TLIN.
\end{proof}

\begin{remark}
Theorem \ref{th:th3} shows that based on the structure of the subproblem, the optimization always has rank one precoders in TSTINR and TLIN as solutions. Simulations verify this behavior in all cases. The same phenomenon is observed for the WMMSE model of \cite{Truong} whereas the conclusion of Theorem
\ref{th:th3} cannot be extended to WMMSE, due to the extra linear term in the objective function of the precoder subproblem. Therefore, with the existing models we can only achieve a single
data stream for each user pair. Thus a new model should be proposed to achieve multiple data streams.
\end{remark}

\subsection{Multiple stream TSTINR model}
In this subsection, we propose the new model based on the TSTINR model in Section \MakeUppercase{\romannumeral 3} to support multiple data streams. Sufficient motivation for the construction of the new model is also provided.

\subsubsection{Analysis of user transmit power allocation}
To achieve the required number of parallel data streams, we should have independent columns of precoder $\mathbf{U}_{k}$ for all $k\in\mathcal {K}$. Without loss of generality we require the columns of $\mathbf{U}_{k}$ to be orthogonal. Whereas there is a transmit power constraint (\ref{eq:eq4.2}) for each user in (\ref{eq:eq4}), we have the power allocation among $d_{k}$ parallel data streams for User $k$. First, we modify our TSTINR model as follows:
\begin{eqnarray} \label{eq:eq39}
\hspace{-0.8cm}&\displaystyle\max_{\substack{\{\mathbf{U}\},\{\mathbf{V}\},\\\{\mathbf{W}\},\{\boldsymbol{\Phi}\}}}& \textrm{TSTINR}=
\frac{\sum_{k\in\mathcal {K}}P_{k}^{S}}{\sum_{k\in\mathcal {K}}(P_{k}^{I}+P_{k}^{N})}\nonumber\\
\hspace{-0.8cm}&\textrm{s.t.}& \mathbf{U}_{k}^{H}\mathbf{U}_{k}=\boldsymbol{\Phi}_{k},\mathbf{V}_{k}^{H}\mathbf{V}_{k}=\mathbf{I}_{d_{k}}, k\in\mathcal {K},\nonumber\\
\hspace{-0.8cm}&&\sum_{r\in\mathcal {R}}\!(\sum_{k\in\mathcal {K}}\!\|\mathbf{W}_{r}\mathbf{G}_{rk}\mathbf{U}_{k}\|_{F}^{2}
+\sigma_{1}^{2}\|\mathbf{W}_{r}\|_{F}^{2})\leq p_{\max}^{R},\nonumber\\
&&\textrm{tr}(\boldsymbol{\Phi}_{k})\leq p_{0}^{T}, \boldsymbol{\Phi}_{k}\hspace{0.1cm}\textrm{is diagonal}, \boldsymbol{\Phi}_{k}\succeq0, k\in\mathcal {K}.
\end{eqnarray}
Here $\boldsymbol{\Phi}_{k}$ is a $d_{k}\times d_{k}$ diagonal positive semi-definite matrix, which contains the data stream power allocation variable of User $k$.

From the optimization point of view, the feasible set of precoder $\mathbf{U}_{k}$ is restricted to have orthogonal columns, comparing with that of (\ref{eq:eq4}). This avoids the phenomenon observed in the solution of (\ref{eq:eq4}) that all columns of the rank one precoders $\mathbf{U}_{k}$ are nonzero but linearly dependent. Different from what has been mentioned in Theorem \ref{th:th3}, here the rank one case of $\mathbf{U}_{k}$ only happens when one diagonal element of $\boldsymbol{\Phi}_{k}$ is nonzero, which result in all columns of $\mathbf{U}_{k}$  but one are all zeros. Hence, we focus on the analysis of the subproblem to solve $\mathbf{U}_{k}$ as well as $\boldsymbol{\Phi}_{k}$. The reformulation of the objective function of (\ref{eq:eq39}) and the update strategy of the parameter $C$ are similar to the algorithm to solve (\ref{eq:eq4}). Given $k\in\mathcal {K}$, fixing all variables other than $\mathbf{U}_{k}$ and $\boldsymbol{\Phi}_{k}$, the precoder subproblem becomes:
\begin{subequations} \label{eq:eq40}
\begin{eqnarray}
&\displaystyle\min_{\mathbf{X}\in\mathbb{C}^{M_{k}\times
d_{k}},\boldsymbol{\Phi}_{k}\in\mathbb{C}^{d_{k}\times d_{k}}} &\textrm{tr}(\mathbf{X}^{H}\mathbf{Q}_{k}\mathbf{X})\\
&\textrm{s.t.}& \mathbf{X}^{H}\mathbf{X}=\boldsymbol{\Phi}_{k},\\
&&\textrm{tr}(\mathbf{X}^{H}\mathbf{L}_{k}\mathbf{X})\leq\eta_{2},\label{eq:eq40.1}\\
&&\textrm{tr}(\boldsymbol{\Phi}_{k})\leq p_{0}^{T}, \boldsymbol{\Phi}_{k}\hspace{0.1cm}\textrm{is diagonal}, \boldsymbol{\Phi}_{k}\succeq0,
\end{eqnarray}
\end{subequations}
where $\mathbf{X}$ represents $\mathbf{U}_{k}$.

\begin{dingli} \label{th:th4}
The optimal $\boldsymbol{\Phi}_{k}$ of (\ref{eq:eq40}) is of rank one, i.e., there is only one positive element on the diagonal of $\boldsymbol{\Phi}_{k}$.
\end{dingli}
The detailed proof is shown in Appendix-C. As described in Theorem \ref{th:th4}, at the optimal solution of (\ref{eq:eq40}) the complete transmit power should be assigned to one data stream. This leads to $\textrm{rank}(\mathbf{U}_{k})=1$ and thus only one data stream can be transmitted for each user.

\subsubsection{New model and the algorithm framework}
In our new model for multiple stream case, without transmit power optimization, we assume each user has fixed transmit power $p_{0}^{T}$, and require equal power allocation among parallel data streams for each user. This choice accords with the optimal power allocation scheme to maximize the system sum rate in the high SNR scenario \cite{Rhee}.
Suppose User $k$ has $d_{k}$ parallel data streams, and the corresponding optimization problem of the new model becomes:
\begin{eqnarray} \label{eq:eq32}
\hspace{-0.8cm}&\displaystyle\max_{\substack{\{\mathbf{U}\},\{\mathbf{V}\},\\\{\mathbf{W}\}}}& \textrm{TSTINR}=
\frac{\sum_{k\in\mathcal {K}}P_{k}^{S}}{\sum_{k\in\mathcal {K}}(P_{k}^{I}+P_{k}^{N})}\nonumber\\
\hspace{-0.8cm}&\textrm{s.t.}& \mathbf{U}_{k}^{H}\mathbf{U}_{k}=\frac{p_{0}^{T}}{d_{k}}\mathbf{I}_{d_{k}},\mathbf{V}_{k}^{H}\mathbf{V}_{k}=\mathbf{I}_{d_{k}}, k\in\mathcal {K},\nonumber\\
\hspace{-0.8cm}&&\sum_{r\in\mathcal {R}}\!(\sum_{k\in\mathcal {K}}\!\|\mathbf{W}_{r}\mathbf{G}_{rk}\mathbf{U}_{k}\|_{F}^{2}
+\sigma_{1}^{2}\|\mathbf{W}_{r}\|_{F}^{2})\leq p_{\max}^{R}.
\end{eqnarray}
Similar to (\ref{eq:eq4}), we reformulate the objective function of (\ref{eq:eq32}) with parameter $C$, which adopts
the update strategy (\ref{eq:eq33}), and becomes $f(\{\mathbf{U}\},\{\mathbf{V}\},\{\mathbf{W}\};C)
=C(P^{I}+P^{N})-P^{S}$. We apply the alternating minimization method to solve precoders, decoders and relay beamforming matrices in the reformulated problem. The subproblems for decoders
$\mathbf{V}_{k}, k\in\mathcal {K}$ are the same as (\ref{eq:eq7}) in Section \ref{sec:sec2}.

For any $r\in\mathcal {R}$, the subproblem for $\mathbf{W}_{r}$ while fixing all other variables is reformulated as (\ref{eq:eq10}) with inequality constraint. Then it is equivalent to the typical trust region subproblem
(\ref{eq:eq17}), and solved by TR method in \cite{Yuan}.

For the precoder $\mathbf{U}_{k}$, the corresponding subproblem
becomes following, while fixing $\mathbf{V}_{q}, q\in\mathcal {K}$, $\mathbf{W}_{r}, r\in\mathcal {R}$ and $\{\mathbf{U}_{-k}\}$:
\begin{subequations} \label{eq:eq35}
\begin{eqnarray}
&\displaystyle\min_{\mathbf{X}\in\mathbb{C}^{M_{k}\times
d_{k}}} &\textrm{tr}(\mathbf{X}^{H}\mathbf{Q}_{k}\mathbf{X})\\
&\textrm{s.t.}& \mathbf{X}^{H}\mathbf{X}=\frac{p_{0}^{T}}{d_{k}}\mathbf{I}_{d_{k}},\\
&&\textrm{tr}(\mathbf{X}^{H}\mathbf{L}_{k}\mathbf{X})\leq\eta_{2},\label{eq:eq35.1}
\end{eqnarray}
\end{subequations}
where $\mathbf{X}$ represents the variable $\mathbf{U}_{k}$, and $\mathbf{Q}_{k}$, $\mathbf{L}_{k}$ and $\eta_{2}$ are mentioned just after the equation of (\ref{eq:eq11}).
Before presenting the algorithm to solve (\ref{eq:eq35}), we first show its global optimality conditions:
\begin{dingli} \label{th:th5}
The global optimality conditions for (\ref{eq:eq35}) are stated as follows.
There exists $\mu^{*}\geq0$ as the Lagrange multiplier of (\ref{eq:eq35.1}), such that:\\
\begin{description}
\item[OC1] $\mathbf{X}^{*}(\mu^{*})=\sqrt{\frac{p_{0}^{T}}{d_{k}}}\nu_{\min}^{d_{k}}(\mathbf{Q}_{k}+\mu^{*}\mathbf{L}_{k})$ is the optimal solution for:
\begin{eqnarray} \label{eq:eq38}
\displaystyle\min_{\mathbf{X}^{H}\mathbf{X}=\frac{p_{0}^{T}}{d_{k}}\mathbf{I}_{d_{k}}} \textrm{tr}[\mathbf{X}^{H}(\mathbf{Q}_{k}+\mu^{*}\mathbf{L}_{k})\mathbf{X}].
\end{eqnarray}
\item[OC2] Complimentary condition holds: $\mu^{*}\{\textrm{tr}[(\mathbf{X}^{*})^{H}\mathbf{L}_{k}\mathbf{X}^{*}]-\eta_{2}\}=0$.
\item[OC3] $c(\mu)$ as the function of $\mu$ satisfies (\ref{eq:eq35.1}):
\begin{center}$c(\mu^{*})=\textrm{tr}\{[\mathbf{X}^{*}(\mu^{*})]^{H}\mathbf{L}_{k}\mathbf{X}^{*}(\mu^{*})\}\leq \eta_{2}$.\end{center}
\end{description}
\end{dingli}
\begin{proof}
Suppose there exists $\mathbf{X}_{0}$ as a feasible point of (\ref{eq:eq35}), which satisfies: $\textrm{tr}(\mathbf{X}_{0}^{H}\mathbf{Q}_{k}\mathbf{X}_{0})<\textrm{tr}[(\mathbf{X}^{*})^{H}\mathbf{Q}_{k}\mathbf{X}^{*}]$.\\
If $\mu^{*}=0$, then from OC1, $\mathbf{X}^{*}$ is the global optimal solution of the relaxed problem of (\ref{eq:eq35}) by dropping the constraint (\ref{eq:eq35.1}). Thus for $\mathbf{X}$ as any feasible point of (\ref{eq:eq35}), it holds that
$\textrm{tr}[(\mathbf{X}^{*})^{H}\mathbf{Q}_{k}\mathbf{X}^{*}]\leq\textrm{tr}(\mathbf{X}^{H}\mathbf{Q}_{k}\mathbf{X})$,
which contradicts the assumption.\\
If $\mu^{*}>0$, then $\textrm{tr}[(\mathbf{X}^{*})^{H}\mathbf{L}_{k}\mathbf{X}^{*}]=\eta_{2}$ holds from OC2. Thus we deduce
$
\mu^{*}\textrm{tr}(\mathbf{X}_{0}^{H}\mathbf{L}_{k}\mathbf{X}_{0})\leq\mu^{*}\eta_{2}=\mu^{*}\textrm{tr}[(\mathbf{X}^{*})^{H}\mathbf{L}_{k}\mathbf{X}^{*}]
$. Then we have
$\textrm{tr}[\mathbf{X}_{0}^{H}(\mathbf{Q}_{k}+\mu^{*}\mathbf{L}_{k})\mathbf{X}_{0}]<
\textrm{tr}[(\mathbf{X}^{*})^{H}(\mathbf{Q}_{k}+\mu^{*}\mathbf{L}_{k})\mathbf{X}^{*}]$,
which contradicts the fact that $\mathbf{X}^{*}$ is the global optimal solution of (\ref{eq:eq38}).
For both cases we have proved that the assumption for $\mathbf{X}_{0}$ does not hold. Thus $\mathbf{X}^{*}$ is the global optimal solution of (\ref{eq:eq35}).
\end{proof}
With eigenvalue decomposition, we obtain the optimal solution of (\ref{eq:eq38}) as $\mathbf{X}^{*}(\mu^{*})=
\sqrt{\frac{p_{0}^{T}}{d_{k}}}\nu_{\min}^{d_{k}}(\mathbf{Q}_{k}+\mu^{*}\mathbf{L}_{k})$. We want to obtain $\mu^{*}\geq0$, to satisfy OC2 and OC3. If $c(0)\leq\eta_{2}$, then $\mu=0$ is the optimal Lagrange multiplier. Otherwise $\mu=0$ cannot be optimal and
$\mu^{*}$ should be strictly greater than $0$. Thus we should always have $c(\mu^{*})=\textrm{tr}[(\mathbf{X}^{*})^{H}\mathbf{L}_{k}\mathbf{X}^{*}]
= \eta_{2}$ from OC2. Also from the constraint (\ref{eq:eq35.1}) we should have $c(\infty)\leq\eta_{2}$ for a feasible problem. Then with
$c(\mu)$ as a continuous function, there exists $\mu^{*}\in(0,\infty)$ such that $c(\mu^{*})=\eta_{2}$. Thus we use Newton's root finding method \cite{Yuan} to
search for $\mu^{*}$.

The algorithm to solve subproblem (\ref{eq:eq35}) is summarized as follows:

\begin{algorithm}
\SetKwInOut{Input}{input}
\SetKwInOut{Output}{output}

\Input{$\mu=0$}
\Output{the optimal solution of (\ref{eq:eq35}): $\mathbf{X}^{*}(\mu^{*})=
\sqrt{\frac{p_{0}^{T}}{d_{k}}}\nu_{\min}^{d_{k}}(\mathbf{Q}_{k}+\mu^{*}\mathbf{L}_{k})$}
\eIf{$c(\mu)\leq\eta_{2}$}{
Set $\mu^{*}=0$\;}{
Apply Newton's root finding method to solve $c(\mu^{*})=\eta_{2}$\;
}
\end{algorithm}

\vspace{5cm}
With the methods for all the three subproblems, the algorithm for (\ref{eq:eq32}) is presented as Algorithm 3:

\begin{algorithm}
\SetKwInOut{Input}{input}
\SetKwInOut{Output}{output}
\caption{Algorithm for multiple stream TSTINR model}

\Input{initial value of $\mathbf{U}_{k},k\in\mathcal {K}$ and $\mathbf{W}_{r},r\in\mathcal {R}$, $C=1$}
\Output{$\mathbf{U}_{k},\mathbf{V}_{k},k\in\mathcal {K}$ and $\mathbf{W}_{r},r\in\mathcal {R}$}
\Repeat{Convergence}{
Update decoder $\mathbf{V}_{k}$ by solving (\ref{eq:eq7}), $k\in\mathcal {K}$\;
Update relay beamforming matrix $\mathbf{W}_{r}$ by solving
(\ref{eq:eq9}) with inequality constraint, $r\in\mathcal {R}$\;
Update precoder $\mathbf{U}_{k}$ by solving (\ref{eq:eq35}), $k\in\mathcal {K}$\;
Update $C$ as $C:=\frac{P^{S}}
{P^{I}+P^{N}}$\;
}
\end{algorithm}

By enforcing the orthogonality constraints to the columns of each precoder, it is guaranteed that $\textrm{rank}(\mathbf{U}_{k})=d_{k}, k\in\mathcal {K}$. Thus each user pair has $d_{k}$ parallel data streams as expected.

\section{Simulations} \label{sec:sec3}
In this section, we evaluate the performances of our proposed algorithms. Simulations include two parts, where the single stream and the multiple streams models are analyzed, respectively. In both parts, each element of $\mathbf{G}_{rk}$ and $\mathbf{H}_{kr},k\in\mathcal {K},r\in\mathcal {R}$ are generated as i.i.d complex Gaussian distribution with zero mean and unit variance. The noise variances are set as $\sigma_{1}^{2}=\sigma_{2}^{2}=\sigma^{2}=1$. Initial values of $\mathbf{U}_{k},k\in\mathcal {K}$ and $\mathbf{W}_{r},r\in\mathcal {R}$ are randomly generated, and scaled to be feasible. Initially, the parameter in TSTINR model is set as $C=1$. For each plotted point, $100$ random realization of different channel\break coefficients are generated to evaluate the average performance.

Here we define SNR as SNR$=\frac{p_{0}^{T}}{\sigma^{2}}=\frac{p_{0}^{R}}{\sigma^{2}}$, and $p_{\max}^{R}=R\cdot p_{0}^{R}$. We use system sum rate $R_{\textrm{sum}}$ as the measure of QoS. We use Matlab simulation server R2010a based on an INTEL® Core i7-875K CPU with 8 GB RAM with the operating system 64 bits Debian Linux to run all the simulations.

\subsection{Single stream models}
In this subsection, we consider a $(4\times2,1)^{4}+2^{4}$ MIMO relay system, which means $M_{k}=4,N_{k}=2,L_{r}=4,d_{k}=1$, for all $k\in\mathcal {K}$ and $r\in\mathcal {R}$, and $K=R=4$.
First, we analyze models with per user and total relay power constraints. Our algorithm for TSTINR model is compared with those of the TLIN model and the WMMSE model with power control in \cite{Truong}. Here we use the Sedumi toolbox to solve the subproblems with SDP algorithm in WMMSE. For the TLIN model, we apply similar algorithm as proposed for TSTINR to speed up.
\begin{figure}[htb]

\begin{minipage}[b]{1.0\linewidth}
  \centering
  \centerline{\includegraphics[width=8cm]{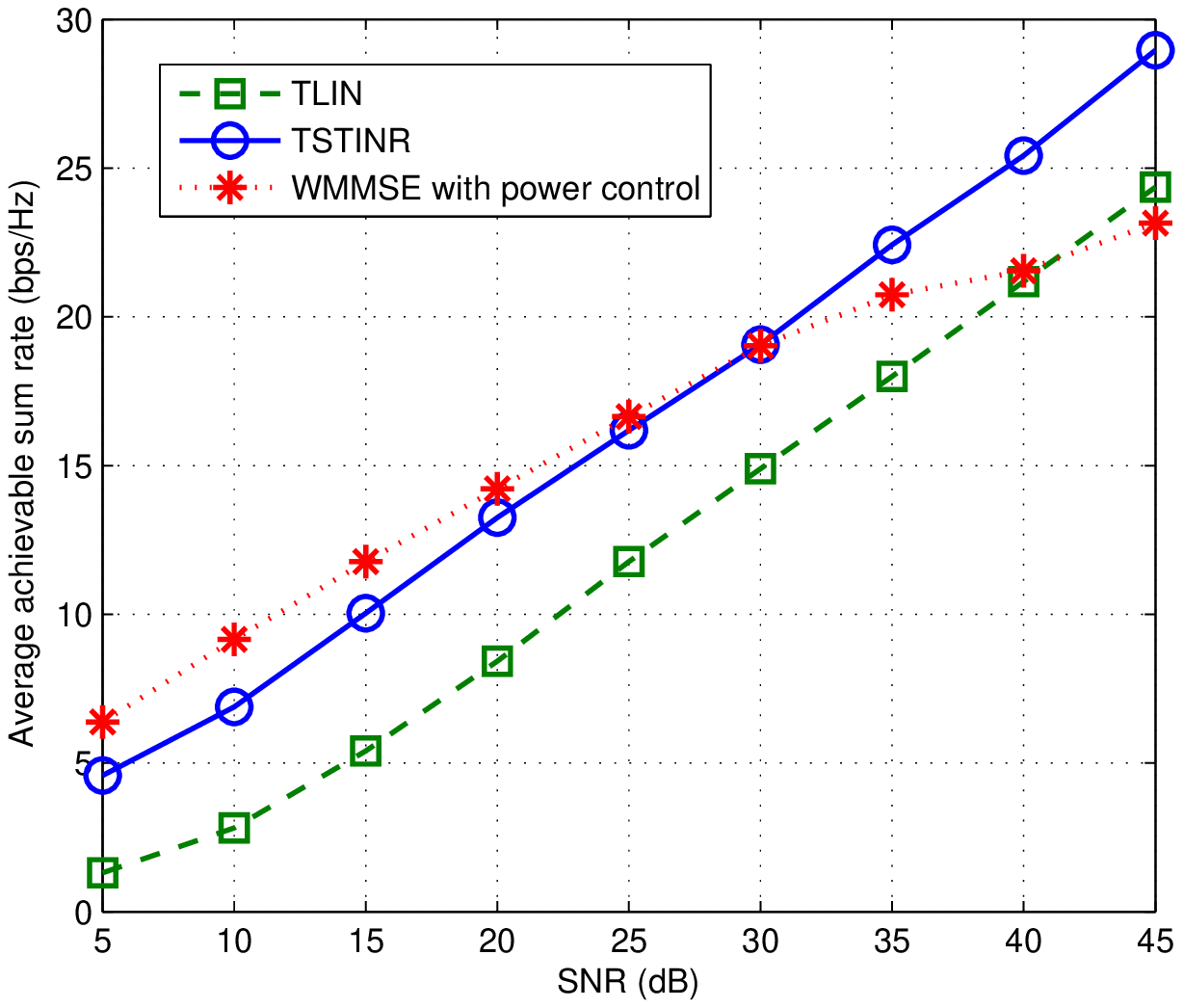}}
  \center{Fig. 2 Average achievable sum rate versus SNR,\\ total relay power constraint}\medskip
\end{minipage}
\begin{minipage}[b]{1.0\linewidth}
  \centering
  \centerline{\includegraphics[width=8cm]{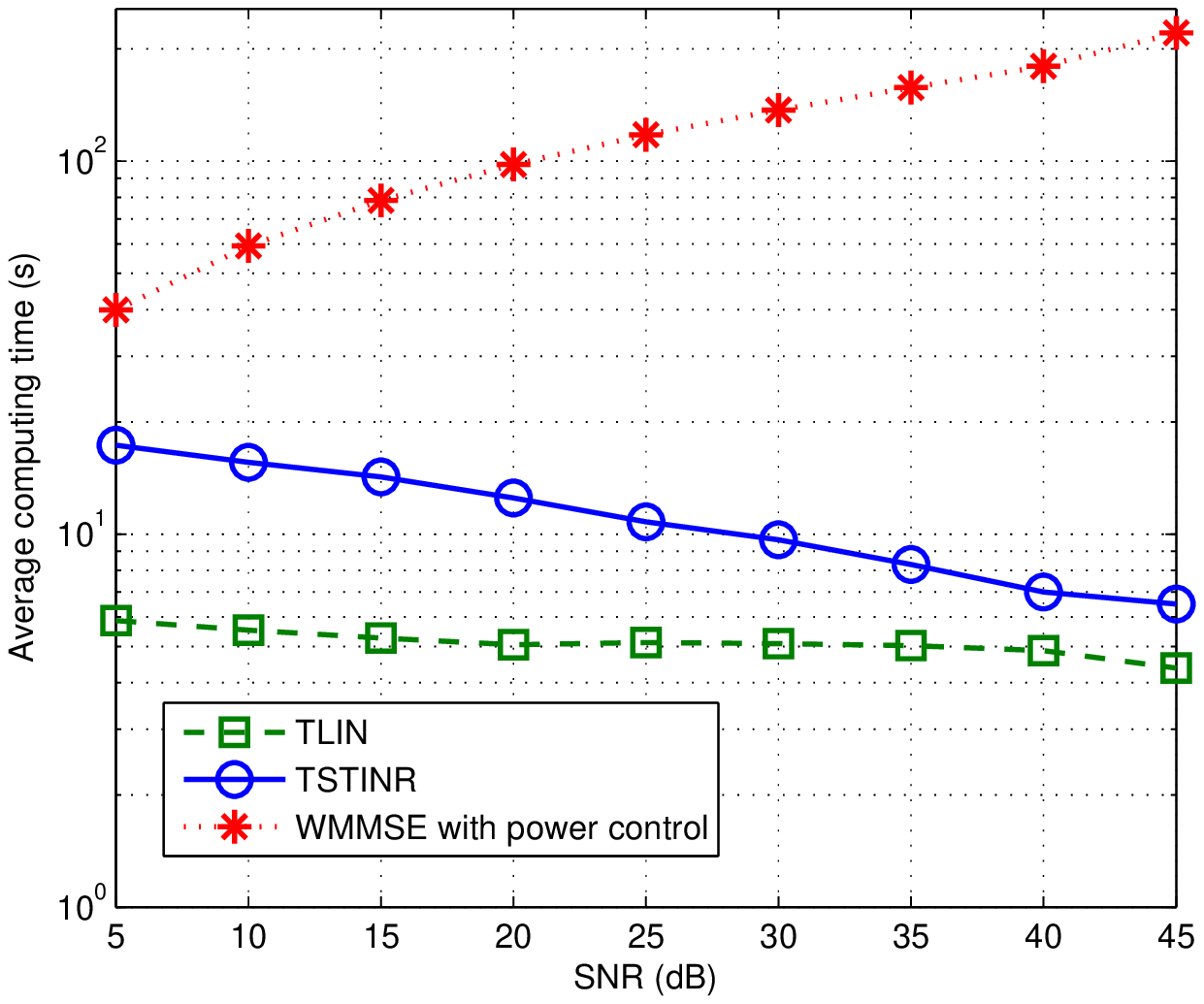}}
  \center{Fig. 3 Average computing time versus SNR,\\ total relay power constraint}\medskip
\end{minipage}
\end{figure}

Fig. 2 and Fig. 3 depict the achieved sum rate and the computing time with respect to different SNR values for the three algorithms, respectively. From the aspect of the achieved sum rate, TSTINR outperforms TLIN for general SNR values as expected, and outperforms WMMSE in medium and high SNR scenarios. Furthermore, the computing time of TSTINR and TLIN are much less than that of WMMSE, which accords with the computational complexity analysis. By numerical evidence, the computing time of WMMSE increases with increasing SNR value, mainly because the rank-one solution is less frequently observed for SDP, and the relaxation technique \cite{Ai} has to be used.

Then we compare the algorithms for models with individual user and individual relay power constraints. Fig. 4 depicts the average achievable sum rate by our proposed algorithms for TLIN, TSTINR and WMMSE model with fixed transmit power, as well as the algorithm in \cite{Truong} for the WMMSE model with power control. It is shown that WMMSE with power control enjoys higher sum rate than WMMSE with fixed power constraints generally. Thus system sum rate benefits substantially from power control. Even so, our proposed TSTINR with fixed power constraints outperforms WMMSE with power control for medium to high SNR scenarios. Also, TSTINR outperforms TLIN in general. Here the complexity as well as the computing time of the TLIN, TSTINR and WMMSE with power control behave similar to the algorithms in Section \ref{sec:sec2}.
\begin{figure}[htb]

\begin{minipage}[b]{1.0\linewidth}
  \centering
  \centerline{\includegraphics[width=8cm]{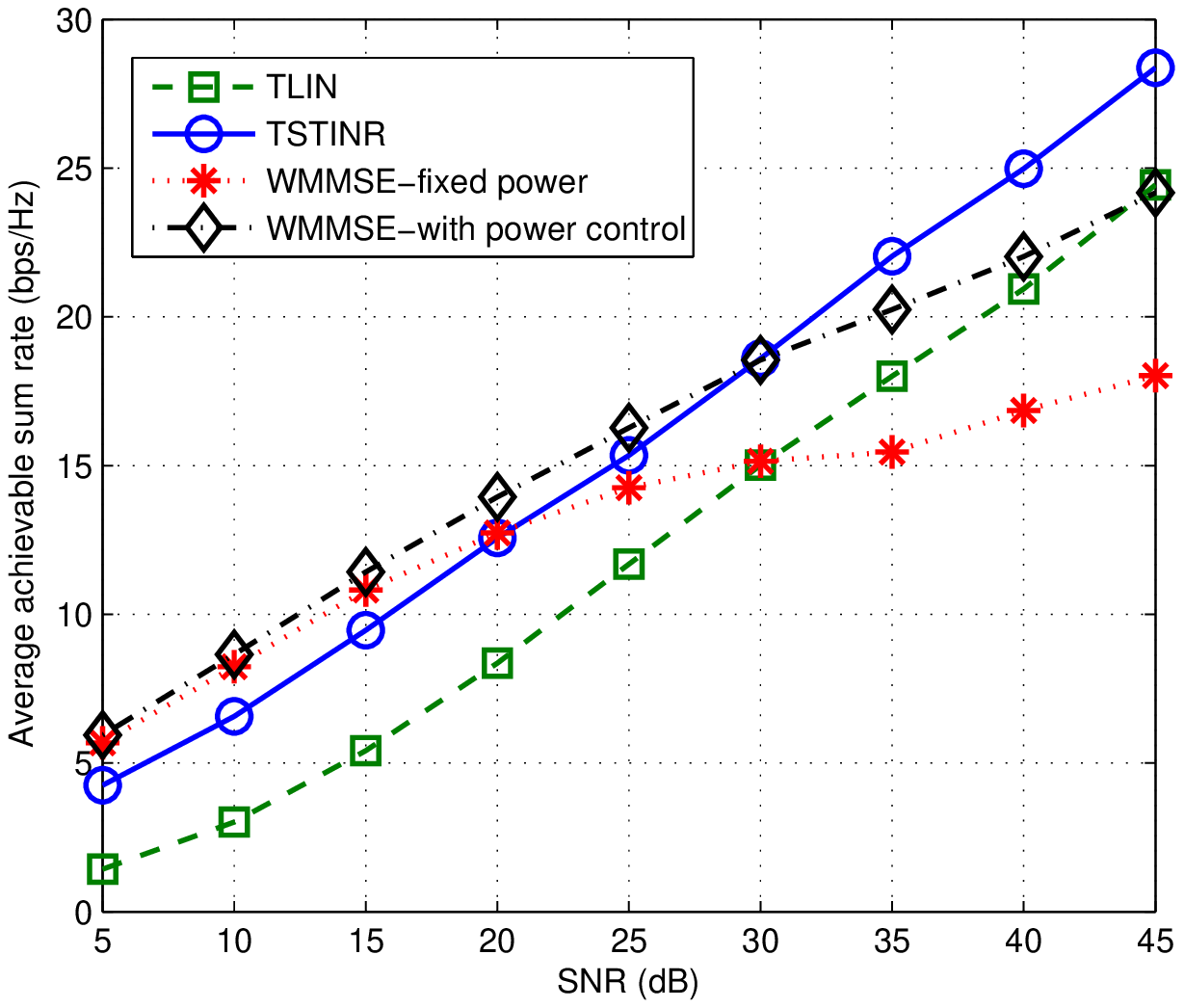}}
  \center{Fig. 4 Average achievable sum rate versus SNR,\\ individual relay power constraint}\medskip
\end{minipage}
\end{figure}

The comparison with both total relay and individual relay constraints cases are similar. In both Fig. 2 and Fig. 4 the curves representing TSTINR and WMMSE with power control have intersection point at about SNR$=30$dB. With SNR value higher than $30$dB, our proposed TSTINR algorithm outperforms the WMMSE algorithm with power control. The performance confirms the analysis: in low SNR scenarios MMSE receiver filter are almost optimal considering linear filter, while in high SNR scenarios the receiver filter should be close to zero-forcing solution, which can be achieved by TSTINR.

\subsection{Multiple stream model}
In this subsection we consider multiple stream models of the MIMO relay networks. We investigate three kinds of $2\times2\times2$ networks with different number of antennas, that is, $K=R=2$.. Here for the scheme $d_{1}=d_{2}=1$ we choose the maximum system sum rate results between our TSTINR model and the WMMSE model with power control in \cite{Truong}.

First we consider a network with $2$ antennas for each user and $4$ antennas for each relay. The number of data streams for User $k$, $d_{k}$, varies from $1$ to $2$ for both $k=1,2$. For different choices of $d_{k}, k=1,2$ the average sum rate results corresponding to different SNR values are shown in Fig. 5. As expected, in the low SNR scenario the single stream scheme with $d_{1}=d_{2}=1$ outperforms other schemes; in medium to high SNR scenarios, the scheme $d_{1}=d_{2}=2$ becomes dominant and the scheme $d_{1}=d_{2}=1$ performs worse than all others, in terms of sum rate.
\begin{figure}[htb]

\begin{minipage}[b]{1.0\linewidth}
  \centering
  \centerline{\includegraphics[width=8cm]{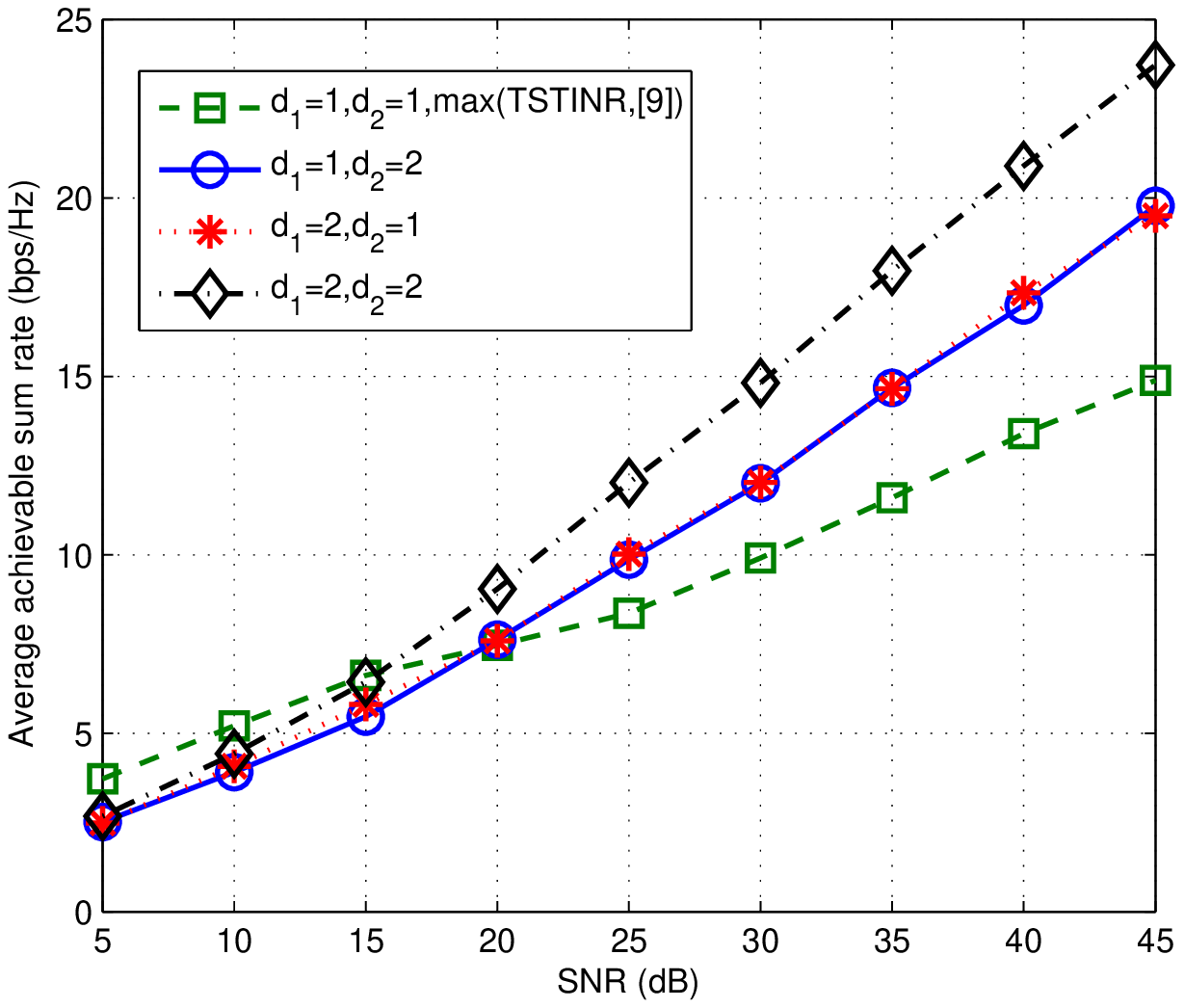}}
  \center{Fig. 5 Average achievable sum rate versus SNR, $M_{k}=N_{k}=2,L_{k}=4$}\medskip
  \end{minipage}
  \begin{minipage}[b]{1.0\linewidth}
  \centering
  \centerline{\includegraphics[width=8cm]{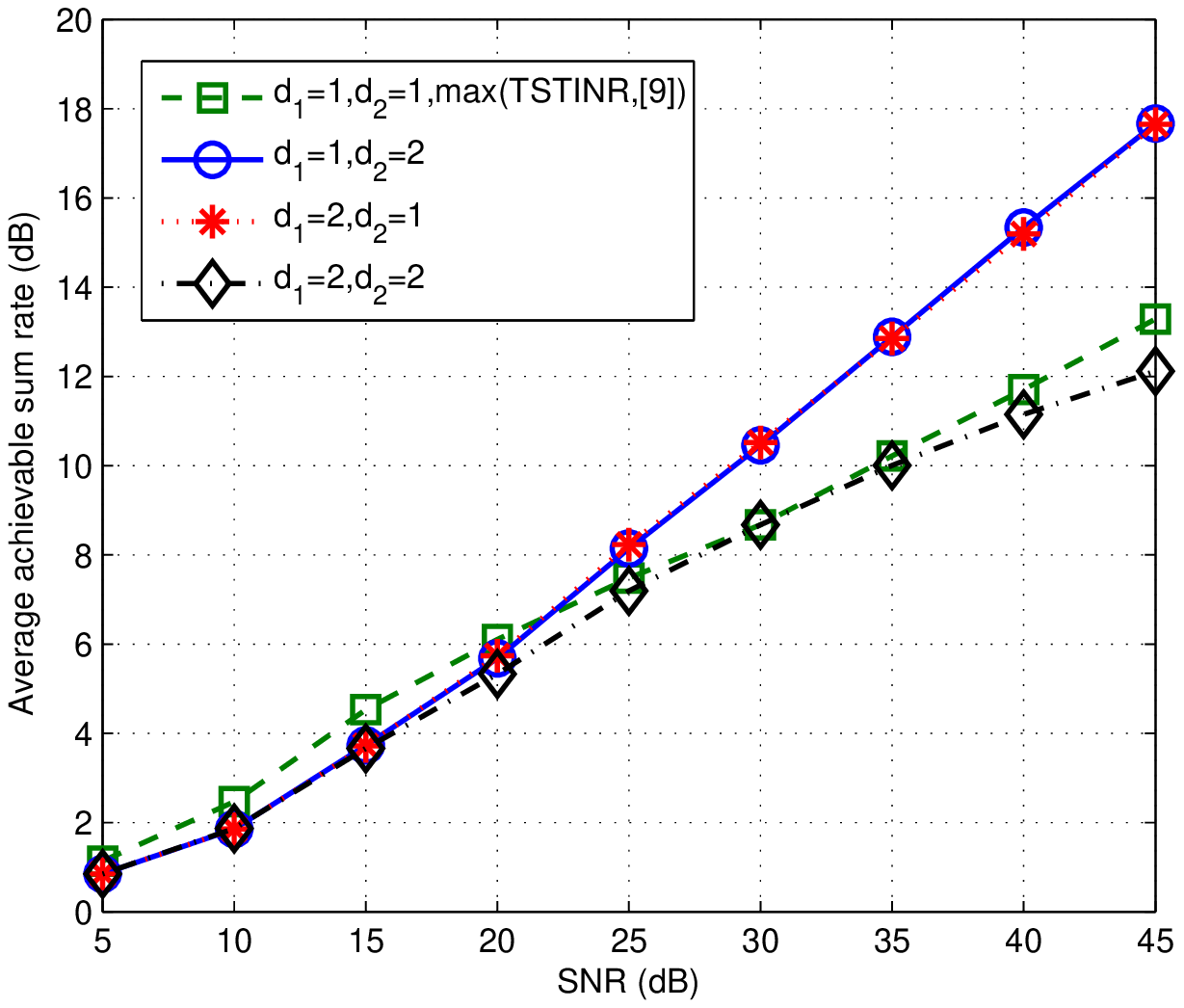}}
  \center{Fig. 6 Average achievable sum rate versus SNR, $M_{k}=N_{k}=L_{k}=2$}\medskip
\end{minipage}
\end{figure}

In the second example, the considered network has the same parameters as the previous one, except that each relay owns $2$ antennas. Similar to Fig. 5, the average achieved sum rate results corresponding to different SNR values for different requirements of data streams are shown in Fig. 6. The curves are quite different from the previous example. Here the schemes $d_{1}=1,d_{2}=2$ and $d_{1}=2,d_{2}=1$ outperform the other two schemes in medium to high SNR scenarios. And generally the scheme $d_{1}=d_{2}=2$ performs very bad.
%

The performances shown in Fig. 5 and Fig. 6 indicate similar behavior as the recent theoretical result on DoF of MIMO relay networks. From the cut-set bound theory, the maximum DoFs of the first network is no greater than $2$ for each user \cite[Theorem 15.1]{Gamal}. And simulations verify the benefit to transmit $2$ data streams for each user over other schemes in Fig. 5. However in the second example, there is no extra relay antenna to align interference besides transmitting the desired signal. Without symbol extension or time division, $2$ DoFs for each user is not achievable. This accords with the performances in Fig. 6. In \cite{Gou} the authors show that $\frac{3}{2}$ DoFs for each user are achievable in $2\times2\times2$ network with $2$ antennas for each user and each relay. And the essential idea of the transmission scheme is to sacrifice one data stream for interference and make full use of all other streams. Correspondingly, in Fig. 6 for the medium to high SNR scenarios the schemes $d_{1}=1,d_{2}=2$ and $d_{1}=2,d_{2}=1$ perform the best. This indicates substantial benefit of system sum rate from such schemes. This is also verified in the third example with $3$ antennas for each user and each relay. The comparison of different data stream schemes are depicted in Fig. 7, where the scheme $d_{1}=2,d_{2}=3$ achieves the highest sum rate among all the schemes in medium to high SNR. In general, multiple stream schemes improve the system sum rate in medium to high SNR scenarios.

\begin{figure}[htb]

\begin{minipage}[b]{1.0\linewidth}
  \centering
  \centerline{\includegraphics[width=8cm]{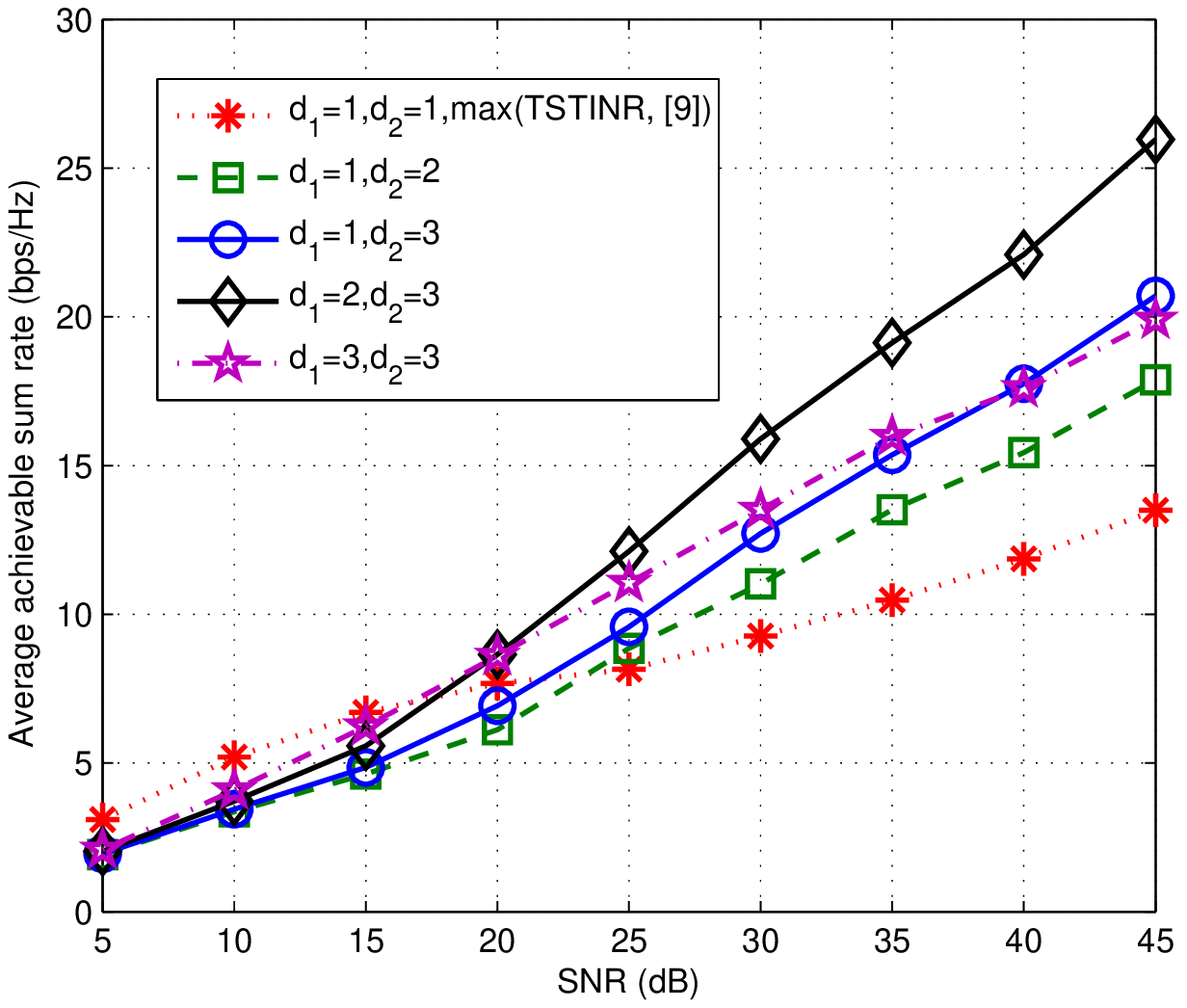}}
  \center{Fig. 7 Average achievable sum rate versus SNR, $M_{k}=N_{k}=L_{k}=3$}\medskip
\end{minipage}
\end{figure}

\section{Conclusion}
This paper considered the general $K\times R\times K$ MIMO relay network. With per user and total relay power constraints, an algorithm for the TSTINR model was proposed. For the individual user and individual relay fixed power constraints, the TSTINR algorithm was extended, and the algorithms for the TLIN and WMMSE model in \cite{Truong} were also modified to solve the problem with such constraints. Computational complexity analysis showed that our proposed algorithm for TSTINR has much lower complexity than the WMMSE algorithm from \cite{Truong}. Focusing on per user and total relay power constraints, we proposed a multiple stream TSTINR model, to overcome the disadvantages of the previous models that only single data stream can be transmitted for each user. Simulations showed that for single stream case TSTINR algorithm performs better than WMMSE in medium and high SNR scenarios for both constraint cases, and outperforms TLIN generally, in terms of achievable sum rate; the system sum rate significantly benefits from the multiple data stream transmission in medium to high SNR scenarios.

\section{Appendix}
\subsection{Proof of Theorem \ref{th:th1}}

First we introduce two lemmas for the proof of Theorem \ref{th:th1}.
\begin{lemma} \cite[Theorem 3.2.2]{Wang}
$$
\mathbf{A}=\left( \begin{array}{ll}
\mathbf{A}_{11} & \mathbf{A}_{12}\\
\mathbf{A}_{21} & \mathbf{A}_{22}
\end{array} \right)\succeq0, \mathbf{B}=\left( \begin{array}{ll}
\mathbf{B}_{11} & \mathbf{B}_{12}\\
\mathbf{B}_{21} & \mathbf{B}_{22}
\end{array} \right)\succeq0
$$
are two Hermitian matrices with the same dimensions. $\mathbf{A}_{11}\succ0$ and $\mathbf{B}_{11}\succ0$ also have the same dimensions. Then
$$
\frac{\det(\mathbf{A}+\mathbf{B})}{\det(\mathbf{A}_{11}+\mathbf{B}_{11})}\geq\frac{\det(\mathbf{A})}{\det(\mathbf{A}_{11})}
+\frac{\det(\mathbf{B})}{\det(\mathbf{B}_{11})}\geq\frac{\det(\mathbf{B})}{\det(\mathbf{B}_{11})}.
$$
So this derives:
$$
\frac{\det(\mathbf{A}+\mathbf{B})}{\det(\mathbf{B})}\geq\frac{\det(\mathbf{A}_{11}+\mathbf{B}_{11})}{\det(\mathbf{B}_{11})}.
$$
\end{lemma}

\begin{lemma} \cite[Theorem 6.8.1]{Wang}
Suppose $\mathbf{C},\mathbf{B}$ are two Hermitian matrices. $\mathbf{C}\succeq\mathbf{B}\succeq0$. Then the following inequality holds:
$$
\frac{\det(\mathbf{C})}{\det(\mathbf{B})}\geq\frac{\textrm{tr}(\mathbf{C})}{\textrm{tr}(\mathbf{B})}.
$$
\end{lemma}

Let $\mathbf{A}_{k}=\mathbf{T}_{kk}\mathbf{T}_{kk}^{H}$, $\mathbf{B}_{k}=\sum_{q\in\mathcal {K},q\neq k}\mathbf{T}_{kq}
\mathbf{T}_{kq}^{H}+\sigma_{1}^{2}\sum_{r\in\mathcal {R}}\bar{\mathbf{H}}_{kr}\bar{\mathbf{H}}_{kr}^{H}+\sigma_{2}^{2}\mathbf{I}_{N}$ and $\mathbf{C}_{k}=
\mathbf{A}_{k}+\mathbf{B}_{k},k\in\mathcal {K}$. Then we have
$$
1+\textrm{TSTINR}=1+\frac{\sum_{k\in\mathcal {K}}\mathbf{V}_{k}^{H}\mathbf{A}_{k}\mathbf{V}_{k}}
{\sum_{k\in\mathcal {K}}\mathbf{V}_{k}^{H}\mathbf{B}_{k}\mathbf{V}_{k}}=\frac{\sum_{k\in\mathcal {K}}\mathbf{V}_{k}^{H}\mathbf{C}_{k}\mathbf{V}_{k}}
{\sum_{k\in\mathcal {K}}\mathbf{V}_{k}^{H}\mathbf{B}_{k}\mathbf{V}_{k}}\textrm{ and }
R_{\textrm{sum}}=\sum_{k\in\mathcal{K}}\textrm{log}_{2}\det(\mathbf{B}_{k}^{-1}\mathbf{C}_{k}).
$$
From the definitions we conclude that $\mathbf{B}_{k}\succ0$, $\mathbf{C}_{k}\succ0$ and $\mathbf{C}_{k}\succeq\mathbf{B}_{k},k\in\mathcal {K}$.

1. We prove that for any $k\in\mathcal {K}$,
\begin{eqnarray} \label{eq:eq23}
\det(\mathbf{B}_{k}^{-1}\mathbf{C}_{k})\geq\det[(\mathbf{V}_{k}^{H}\mathbf{B}_{k}\mathbf{V}_{k})^{-1}(\mathbf{V}_{k}^{H}\mathbf{C}_{k}\mathbf{V}_{k})].
\end{eqnarray}
For simplicity we omit the subscript $k$. Let $\mathbf{V}_{\perp}\in\mathbb{C}^{N\times(N-d)}$ be the bases of the complementary subspace of the subspace
spanned by the columns of
$\mathbf{V}$. That is, $\mathbf{Q}=[\mathbf{V},\mathbf{V}_{\perp}]\in\mathbb{C}^{N\times N}$ is a unitary matrix. Then we deduce that
$$
\det(\mathbf{B}^{-1}\mathbf{C})=\frac{\det(\mathbf{Q}^{H}\mathbf{C}\mathbf{Q})}{\det(\mathbf{Q}^{H}\mathbf{B}\mathbf{Q})}\geq
\frac{\det(\mathbf{V}^{H}\mathbf{C}\mathbf{V})}{\det(\mathbf{V}^{H}\mathbf{B}\mathbf{V})}.
$$
Here the equality is deduced from the property of unitary matrix $\mathbf{Q}$. And the inequality comes from Lemma 1,
$\mathbf{C}=\mathbf{A}+\mathbf{B}$ and that
$$
\mathbf{Q}^{H}\mathbf{C}\mathbf{Q}=\left( \begin{array}{ll}
\mathbf{V}^{H}\mathbf{C}\mathbf{V} & \mathbf{V}^{H}\mathbf{C}\mathbf{V}_{\perp}\\
\mathbf{V}_{\perp}^{H}\mathbf{C}\mathbf{V} & \mathbf{V}_{\perp}^{H}\mathbf{C}\mathbf{V}_{\perp}
\end{array} \right).
$$

2. $\mathbf{C}_{k}\succeq\mathbf{B}_{k}$ induces $\mathbf{V}_{k}^{H}\mathbf{C}_{k}\mathbf{V}_{k}\succeq\mathbf{V}_{k}^{H}\mathbf{B}_{k}\mathbf{V}_{k}$.
Lemma 2 shows that
\begin{eqnarray} \label{eq:eq25}
\frac{\det(\mathbf{V}_{k}^{H}\mathbf{C}_{k}\mathbf{V}_{k})}{\det(\mathbf{V}_{k}^{H}\mathbf{B}_{k}\mathbf{V}_{k})}\geq
\frac{\textrm{tr}(\mathbf{V}_{k}^{H}\mathbf{C}_{k}\mathbf{V}_{k})}
{\textrm{tr}(\mathbf{V}_{k}^{H}\mathbf{B}_{k}\mathbf{V}_{k})}.
\end{eqnarray}
From (\ref{eq:eq23}) and (\ref{eq:eq25}), it is concluded that
\begin{eqnarray} \label{eq:eq26}
R_{\textrm{sum}}=\sum_{k\in\mathcal{K}}\textrm{log}_{2}\det(\mathbf{B}_{k}^{-1}\mathbf{C}_{k})\geq\sum_{k\in\mathcal{K}}\textrm{log}_{2}
\det[(\mathbf{V}_{k}^{H}\mathbf{B}_{k}\mathbf{V}_{k})^{-1}(\mathbf{V}_{k}^{H}\mathbf{C}_{k}\mathbf{V}_{k})]\geq
\sum_{k\in\mathcal{K}}\textrm{log}_{2}\frac{\textrm{tr}(\mathbf{V}_{k}^{H}\mathbf{C}_{k}\mathbf{V}_{k})}
{\textrm{tr}(\mathbf{V}_{k}^{H}\mathbf{B}_{k}\mathbf{V}_{k})}.
\end{eqnarray}

3. Finally we prove that:
\begin{eqnarray} \label{eq:eq22}
\sum_{k\in\mathcal{K}}\textrm{log}_{2}\frac{\textrm{tr}(\mathbf{V}_{k}^{H}\mathbf{C}_{k}\mathbf{V}_{k})}
{\textrm{tr}(\mathbf{V}_{k}^{H}\mathbf{B}_{k}\mathbf{V}_{k})}\geq\textrm{log}_{2}\frac{\sum_{k\in\mathcal{K}}\textrm{tr}(\mathbf{V}_{k}^{H}\mathbf{C}_{k}
\mathbf{V}_{k})}{\sum_{k\in\mathcal{K}}\textrm{tr}(\mathbf{V}_{k}^{H}\mathbf{B}_{k}\mathbf{V}_{k})}.
\end{eqnarray}
With any scalar $t_{k}\geq1$ and the fact that $\textrm{tr}(\mathbf{V}_{k}^{H}\mathbf{B}_{k}\mathbf{V}_{k})\geq0, k\in\mathcal{K}$, it is deduced that:
$$
(\prod_{k\in\mathcal{K}}t_{k})\sum_{k\in\mathcal{K}}\textrm{tr}(\mathbf{V}_{k}^{H}\mathbf{B}_{k}\mathbf{V}_{k})\geq
\sum_{k\in\mathcal{K}}t_{k}\textrm{tr}(\mathbf{V}_{k}^{H}\mathbf{B}_{k}\mathbf{V}_{k}).
$$
Let $t_{k}=\frac{\textrm{tr}(\mathbf{V}_{k}^{H}\mathbf{C}_{k}\mathbf{V}_{k})}
{\textrm{tr}(\mathbf{V}_{k}^{H}\mathbf{B}_{k}\mathbf{V}_{k})}, k\in\mathcal{K}$, divide $\sum_{k\in\mathcal{K}}\textrm{tr}(\mathbf{V}_{k}^{H}\mathbf{B}_{k}
\mathbf{V}_{k})>0$ and take logarithm for both sides, and thus we have (\ref{eq:eq22}). Combining (\ref{eq:eq26}) and (\ref{eq:eq22}) we prove Theorem \ref{th:th1}.

\subsection{Proof of Theorem \ref{th:th2}}

Let $\{\mathbf{X}\}$ represent the set of the iterative points $\{\{\mathbf{U}\},\{\mathbf{V}\},\{\mathbf{W}\}\}$. Suppose $\{\mathbf{X}^{i}\}$ are the
feasible points achieved from the $i$th iteration. Define $P^{I+N}=P^{I}+P^{N}$. Then the expression of parameter $C$ used in the $i$th iteration as well as the TSTINR achieved in the $(i-1)$th iteration is:
$$
C^{i}=\textrm{TSTINR}^{i-1}=\frac{P^{S}(\{\mathbf{X}^{i-1}\})}
{P^{I+N}(\{\mathbf{X}^{i-1}\})}.
$$

As in the $i$th iteration there is sufficient reduction of (\ref{eq:eq6.1}), it holds that
$$
f(\{\mathbf{X}^{i}\};C^{i})=C^{i}P^{I+N}(\{\mathbf{X}^{i}\})-P^{S}(\{\mathbf{X}^{i}\})\leq f(\{\mathbf{X}^{i-1}\};C^{i})=0,
$$
Then $\textrm{TSTINR}^{i}=\frac{P^{S}(\{\mathbf{X}^{i}\})}
{P^{I+N}(\{\mathbf{X}^{i}\})}\geq C^{i}=\textrm{TSTINR}^{i-1}$. Thus the value of TSTINR increases monotonically.

Suppose $\{\mathbf{X}^{*}\}\triangleq\{\{\mathbf{U}^{*}\},\{\mathbf{V}^{*}\},\{\mathbf{W}^{*}\}\}$ are the stationary points of (\ref{eq:eq6})
and $\lambda\in\mathbb{R}$ is the Lagrange multiplier of (\ref{eq:eq6.4}): $h(\{\mathbf{X}\})
=\sum_{r\in\mathcal {R}}(\sum_{k\in\mathcal {K}}\|\mathbf{W}_{r}\mathbf{G}_{rk}\mathbf{U}_{k}\|_{F}^{2}
+\sigma_{1}^{2}\|\mathbf{W}_{r}\|_{F}^{2})-p_{\max}^{R}=0$. Then the first order optimality conditions of the problem (\ref{eq:eq6}) with respect to $\mathbf{W}_{r}$ are:
\begin{eqnarray}
C\frac{\partial P^{I+N}(\{\mathbf{X}^{*}\})}{\partial\mathbf{W}_{r}}-\frac{\partial P^{S}(\{\mathbf{X}^{*}\})}{\partial\mathbf{W}_{r}}
-\lambda\frac{\partial h(\{\mathbf{X}^{*}\})}{\partial\mathbf{W}_{r}}=\mathbf{0};\label{eq:eq27}\\
h(\{\mathbf{X}^{*}\})=0.\label{eq:eq36}
\end{eqnarray}
When the iterative points converges to $\{\mathbf{X}^{*}\}$, we have $C=\frac{P^{S}(\{\mathbf{X}^{*}\})}{P^{I+N}(\{\mathbf{X}^{*}\})}$. Taking it into (\ref{eq:eq27}), and let
$\tilde{\lambda}=-\frac{1}{P^{I+N}(\{\mathbf{X}^{*}\})}\lambda$. Then we have
\begin{eqnarray}  \label{eq:eq37}
\frac{1}{[P^{I+N}(\{\mathbf{X}^{*}\})]^2}\left[P^{I+N}(\{\mathbf{X}^{*}\})\frac{\partial P^{S}(\{\mathbf{X}^{*}\})}{\partial\mathbf{W}_{r}}
-P^{S}(\{\mathbf{X}^{*}\})\frac{\partial P^{I+N}(\{\mathbf{X}^{*}\})}{\partial\mathbf{W}_{r}}\right]
-\tilde{\lambda}\frac{\partial h(\{\mathbf{X}^{*}\})}{\partial\mathbf{W}_{r}}=\mathbf{0}.
\end{eqnarray}
With $\tilde{\lambda}$ as the Lagrange multiplier of (\ref{eq:eq4.1}), (\ref{eq:eq37}) and (\ref{eq:eq36}) consist of the first order optimality conditions of (\ref{eq:eq4}) with respect to $\mathbf{W}_{r}$. Similarly, we are able to achieve the first order optimality
 conditions of (\ref{eq:eq4}) with respect to other variables. Thus $\{\mathbf{X}^{*}\}$ are also the stationary points of (\ref{eq:eq4}).

 \subsection{Proof of Theorem \ref{th:th4}}
Similar to Theorem \ref{th:th5} for problem (\ref{eq:eq35}), the optimality conditions for (\ref{eq:eq40}) are as follows:\\
1. $\mathbf{X}^{*}(\theta^{*}),\boldsymbol{\Phi}_{k}^{*}(\theta^{*})$ are the optimal solutions for the problem below:
\begin{eqnarray} \label{eq:eq8}
&\displaystyle\min_{\mathbf{X}\in\mathbb{C}^{M_{k}\times d_{k}},\boldsymbol{\Phi}_{k}\in\mathbb{C}^{d_{k}\times d_{k}}}& \textrm{tr}[\mathbf{X}^{H}(\mathbf{Q}_{k}+\theta^{*}\mathbf{L}_{k})\mathbf{X}]\nonumber\\
&\textrm{s.t.}& \mathbf{X}^{H}\mathbf{X}=\boldsymbol{\Phi}_{k},\nonumber\\
&&\textrm{tr}(\boldsymbol{\Phi}_{k})\leq p_{0}^{T}, \boldsymbol{\Phi}_{k}\hspace{0.1cm}\textrm{is diagonal}, \boldsymbol{\Phi}_{k}\succeq0.
\end{eqnarray}
2. Complimentary condition holds: $\theta^{*}\{\textrm{tr}[(\mathbf{X}^{*})^{H}\mathbf{L}_{k}\mathbf{X}^{*}]-\eta_{2}\}=0$.\\
3. $c(\theta^{*})$ satisfies (\ref{eq:eq40.1}): $c(\theta^{*})=\textrm{tr}\{[\mathbf{X}^{*}(\theta^{*})]^{H}\mathbf{L}_{k}\mathbf{X}^{*}\}\leq \eta_{2}$.

Here $\theta^{*}\geq0$ is the optimal Lagrange multiplier of (\ref{eq:eq40.1}). With $\theta^{*}$ we are able to obtain the optimal $\mathbf{X}^{*}$ and $\boldsymbol{\Phi}_{k}^{*}$ of (\ref{eq:eq40}).
If the $i$th diagonal element of $\boldsymbol{\Phi}_{k}$ is zero, then the $i$th column of $\mathbf{X}$ should be $\mathbf{0}$. In this situation we can delete this column and optimize the remaining ones. Without loss of generality we assume each element of the diagonal of $\boldsymbol{\Phi}_{k}$ is strictly positive, that is, $\boldsymbol{\Phi}_{k}\succ0$. Let $\mathbf{R}=\mathbf{Q}_{k}+\mu^{*}\mathbf{L}_{k}$, and $\mathbf{Y}=\mathbf{X}\boldsymbol{\Phi}_{k}^{-\frac{1}{2}}$. Omit the index $k$ for simplicity.
Problem (\ref{eq:eq8}) is equivalent to the following problem:
\begin{subequations}\label{eq:eq30}
\begin{eqnarray}
&\displaystyle\min_{\mathbf{Y}\in\mathbb{C}^{M\times d},\boldsymbol{\Phi}\in\mathbb{C}^{d\times d}}& \textrm{tr}(\mathbf{Y}^{H}\mathbf{RY}\boldsymbol{\Phi})\\
&\textrm{s.t.}& \mathbf{Y}^{H}\mathbf{Y}=\mathbf{I}_{d},\label{eq:eq30.1}\\
&&\textrm{tr}(\boldsymbol{\Phi})\leq p_{0}^{T},\boldsymbol{\Phi}\hspace{0.1cm}\textrm{is diagonal}, \boldsymbol{\Phi}\succeq0.
\end{eqnarray}
\end{subequations}
In the following we analyze the optimal solutions of (\ref{eq:eq30}) to show the property of $\mathbf{Y}^{*}$ and $\boldsymbol{\Phi}^{*}$. Here we treat the columns of $\mathbf{Y}$ as the linear combination of the eigenvectors of $\mathbf{R}$. Let $t_{1}\leq t_{2} \leq\ldots\leq t_{M}$ be the eigenvalues of $\mathbf{R}$. Suppose $\mathbf{y}_{i}$, the $i$th column of $\mathbf{Y}$ is the linear combination of the eigenvectors corresponding to the eigenvalues $\{t_{j},j\in\Omega_{i}\}$, where $\Omega_{i}\subseteq\{1,2,\ldots,M\}$. Then according to (\ref{eq:eq30.1}), different columns correspond to different eigenvalues of $\mathbf{R}$, that is, $\cap_{i=1}^{d}\Omega_{i}=\emptyset$. Then the feasible set is divided into $d$ independent parts with $\mathbf{y}_{i}$ as variables, for $i=1,2,\ldots,d$, respectively. Define $\Phi_{ii}$ as the $i$th diagonal element of $\boldsymbol{\Phi}$. The objective function of (\ref{eq:eq30}) is rewritten as $\sum_{i=1}^{d}\Phi_{ii}\mathbf{y}_{i}^{H}\mathbf{Ry}_{i}$. For any fixed feasible $\boldsymbol{\Phi}\succeq0$, the optimal solution of (\ref{eq:eq30}), $\mathbf{y}_{i}^{*}$, is the eigenvector of matrix $\mathbf{R}$ corresponding to the eigenvalue $\bar{t}_{i}\triangleq\min\{t_{j},j\in\Omega_{i}\}$, as long as $\Phi_{ii}>0$, for any $i=1,2,\ldots,d$. Then taking these solutions back into (\ref{eq:eq30}), the problem becomes:
\begin{eqnarray} \label{eq:eq41}
&\displaystyle\min_{\bar{t}_{i},\Phi_{ii},i=1,\ldots,d}& \sum_{i=1}^{d}\Phi_{ii}\bar{t}_{i}\nonumber\\
&\textrm{s.t.}& \Phi_{ii}\geq0, i=1,\ldots,d,\nonumber\\
&& \sum_{i=1}^{d}\Phi_{ii}\leq p_{0}^{T},\nonumber\\
&& \bar{t}_{i}\hspace{0.1cm}\textrm{is an eigenvalue of $\mathbf{R}$}, \bar{t}_{i}\neq\bar{t}_{j}, i,j=1,\ldots,d.
\end{eqnarray}
As the objective function of (\ref{eq:eq41}) is linear in $\Phi_{ii}$ for any $i=1,\ldots,d$, the optimal solutions of (\ref{eq:eq41}) should be
$\Phi_{11}^{*}=p_{0}^{T}, \Phi_{ii}^{*}=0, i=2,\ldots,d$ and $\bar{t}_{1}^{*}=t_{1}$.
Thus, the optimal solution $\boldsymbol{\Phi}_{k}^{*}$ of (\ref{eq:eq40}) is of rank one.

\end{document}